\documentclass[a4paper,english,numberwithinsect, cleveref, autoref]{lipics-v2019}
\usepackage[noabbrev,capitalise]{cleveref}
\usepackage{tikz}%

\usepackage{xspace}

\usepackage{cite}

\hideLIPIcs

\usetikzlibrary{calc}
\usetikzlibrary{arrows.meta}
\usetikzlibrary{decorations.pathreplacing, angles, quotes, decorations.pathmorphing}
\usetikzlibrary{decorations.markings}
\usetikzlibrary{through}
\usetikzlibrary{patterns}
\usetikzlibrary{intersections}
\usetikzlibrary{backgrounds}
\usetikzlibrary{3d}
\usetikzlibrary{positioning}
\usetikzlibrary{shapes}

\tikzstyle{vertex}=[%
    draw,%
    circle,%
    minimum width=0.1cm,%
    fill,%
    outer sep=0,%
    inner sep=0,%
    ]

\graphicspath{{figs/}}

\newcommand\C{{\mathcal C}}
\newcommand\cD{{\mathcal I}}
\newcommand\E{{\mathcal E}}

\newcommand\G{{\mathcal G}}

\def\N{{\mathbb N}}

\def\Q{{\mathcal Q}}

\newcommand\cL{{\mathcal L}}

\def\T{{\mathcal T}}

\newcommand{\D}{D}

\newcommand{\mathset}[1]{\ensuremath {\mathbb {#1}}}
\newcommand{\eps}{\varepsilon}
\newcommand{\R}{\mathset{R}}

\newcommand{\Lacki}{{\L}{\c a}cki\xspace}

\newcommand{\typeROne}{R1}
\newcommand{\typeRTwo}{R2}

\definecolor{ourblue}{RGB}{0,0,255}
\definecolor{ourlightgreen}{RGB}{0,255,0}
\definecolor{ourred}{RGB}{255,0,0}
\definecolor{ourdarkgreen}{RGB}{0,100,0}
\definecolor{ouryellow}{RGB}{255,165,0}

\title{Triangles and Girth in Disk Graphs and 
Transmission Graphs}

\titlerunning{Triangles and Girth in Disk Graphs and Transmission Graphs}
\author{Haim Kaplan}%
{School of Computer Science, Tel Aviv University, Tel~Aviv 69978,
   Israel}%
{haimk@tau.ac.il}%
{}%
{}%
\author{Katharina Klost}%
{Institut f\"ur Informatik, Freie Universit\"at Berlin, 14195 Berlin,
   Germany} {kathklost@inf.fu-berlin.de}%
{}%
{}
\author{Wolfgang Mulzer}%
{Institut f\"ur Informatik, Freie Universit\"at Berlin, 14195 Berlin,
   Germany} {mulzer@inf.fu-berlin.de}%
{0000-0002-1948-5840}%
{Partially supported by ERC STG 757609.}
\author{Liam Roditty}%
{Department of Computer Science, Bar Ilan University, Ramat Gan
   5290002, Israel}%
{liamr@macs.biu.ac.il}%
{}%
{}
\author{Paul Seiferth}%
{Institut f\"ur Informatik, Freie Universit\"at Berlin, 14195 Berlin,
   Germany}%
{pseiferth@inf.fu-berlin.de}%
{}%
{Partially supported by DFG grant MU/3501/1.}
\author{Micha Sharir}%
{School of Computer Science, Tel Aviv University, Tel~Aviv 69978,
   Israel}%
{michas@tau.ac.il}%
{}%
{Partially supported by ISF Grant 892/13, by the Israeli Centers of
   Research Excellence (I-CORE) program (Center No.~4/11), by the
   Blavatnik Research Fund in Computer Science at Tel Aviv University,
   and by the Hermann Minkowski-MINERVA Center for Geometry at Tel
   Aviv University.}

\authorrunning{H. Kaplan, K. Klost, W. Mulzer, L. Roditty,
   P. Seiferth, M. Sharir}

\Copyright{Haim Kaplan, Katharina Klost, Wolfgang Mulzer, Liam
   Roditty, Paul Seiferth, Micha Sharir}



\funding{Work on this paper was supported in part by grant 1367/2016
   from the German-Israeli Science Foundation (GIF).}

\acknowledgements{We like to thank G\"unther Rote and 
Valentin Polishchuk for helpful comments.}

\begin{document}
\nolinenumbers
\maketitle
\begin{abstract}
Let $S \subset \R^2$ be a set of $n$
\emph{sites}, where each $s \in S$ has 
an \emph{associated radius} $r_s > 0$. 
The \emph{disk graph} $\D(S)$ is the 
undirected graph with vertex set $S$ 
and an undirected edge between two 
sites $s, t \in S$ if and only if 
$|st| \leq r_s + r_t$, i.e., if the 
disks with centers $s$ and $t$ and 
respective radii $r_s$ and $r_t$ 
intersect.  Disk graphs are used to 
model sensor networks. Similarly, the 
\emph{transmission graph} $T(S)$ is 
the directed graph with vertex set 
$S$ and a directed edge from a site 
$s$ to a site $t$ if and only if 
$|st| \leq r_s$, i.e., if $t$ lies 
in the disk with center $s$ and 
radius $r_s$.

We provide algorithms for detecting 
(directed) triangles and, more 
generally, computing the length of 
a shortest cycle (the \emph{girth}) 
in $\D(S)$ and in $T(S)$. These 
problems are notoriously hard in 
general, but better solutions exist 
for special graph classes such as 
planar graphs. We obtain similarly 
efficient results for disk graphs 
and for transmission graphs. More 
precisely, we show that a shortest 
(Euclidean) triangle in $\D(S)$ and 
in $T(S)$ can be found in 
$O(n \log n)$ expected time, and that 
the (weighted) girth of $\D(S)$ can 
be found in $O(n \log n)$ expected 
time. For this, we develop new tools 
for batched range searching that may
be of independent interest.
\end{abstract}

\section{Introduction}

Given a graph $G$ with $n$ 
vertices and $m$ edges, does $G$ contain 
a \emph{triangle} (a 
cycle with three vertices)?
This is one of the most basic 
algorithmic questions in graph 
theory, and many other problems 
reduce to it~\cite{ItaiRo78,WiWi18}. The best known algorithms 
use fast matrix multiplication 
and run in either $O(n^{\omega})$ 
time or in $O\big(m^{2\omega/(\omega+1)}\big)$ 
time, where $\omega < 2.37287$ is the 
matrix multiplication 
exponent~\cite{AlonYuZw97,LeGall14,ItaiRo78}.
Despite decades of research, the best available
``combinatorial'' algorithm\footnote{An
algorithm is ``combinatorial'' if it 
does not need algebraic
manipulations to achieve its goal.}
needs 
$O\big(n^3\,\text{polyloglog}(n)/\log^{4}n\big)$ 
time~\cite{Yu15}, only slightly 
better than checking all vertex triples. This lack 
of progress can be explained by a connection 
to Boolean matrix 
multiplication (BMM): 
if 
there is a truly subcubic combinatorial 
algorithm for finding triangles, 
there is also a truly subcubic combinatorial 
algorithm for BMM~\cite{WiWi18}.
Itai and Rodeh~\cite{ItaiRo78} 
reduced  computing the 
\emph{girth} (the length of a shortest 
cycle) of an unweighted undirected 
graph to triangle detection.
For integer edge weights, 
Roditty and V.~Williams~\cite{RodittyVa11} gave 
an equivalence between finding 
a minimum weight cycle (the weighted girth) and 
finding a minimum weight triangle.

For the special case of \emph{planar} graphs,
significantly better algorithms are known. 
Itai and Rodeh~\cite{ItaiRo78} 
and, independently, Papadimitriou and 
Yannakakis~\cite{PapadimitriouYa81} showed that a 
triangle can be found in $O(n)$ time, if it exists. 
Chang and Lu~\cite{ChangLu13} presented 
an $O(n)$ time algorithm for computing the girth.
The weighted girth 
can be found in $O(n \log\log n)$ time 
both in an undirected
and in a directed planar 
graph~\cite{lacki_min-cuts_2011,MozesNiNuWe18}.

In computational geometry, there are two noteworthy 
graph classes that generalize planar graphs:
\emph{disk graphs} and 
\emph{transmission graphs}. We 
are given a set $S$ of $n$ planar point \emph{sites}.
Each $s \in S$ has an 
\emph{associated radius} $r_s > 0$ and an 
\emph{associated disk} $D_s$ with center $s$ 
and radius $r_s$.  The \emph{disk 
graph} $\D(S)$ is the undirected graph 
on $S$ where two sites $s, t \in S$ are adjacent 
if and only if $D_s$ and $D_t$ intersect, i.e.,
$|st| \leq r_s + r_t$, where $|\cdot|$ is the
Euclidean distance.
In a \emph{weighted disk graph}, the edges 
are weighted according to the Euclidean 
distance between their endpoints. The 
\emph{transmission graph} $T(S)$ is the 
directed graph on $S$ where there is an edge from $s$ 
to $t$ if and only if $t$ lies in 
$D_s$, i.e., $|st| \leq r_s$. 
Again, there is a weighted variant.
Both graph classes
have received a lot of attention,
as 
they give simple and natural theoretical models for geometric 
sensor networks (see, e.g.,~\cite{KaplanMuRoSe15,KaplanMuRoSe18}).

Motivated by the vastly better algorithms for 
planar graphs, we investigate triangle detection and
girth computation in disk graphs and transmission graphs.
We will see that in a disk graph, a triangle
can be found in $O(n \log n)$ time, using
a simple geometric
observation to relate disk graphs and planar
graphs. 
By a reduction from \textsc{$\eps$-closeness}~\cite{Polishchuk17},
this is optimal in the
algebraic decision tree model,
a contrast to planar graphs, where $O(n)$ time is 
possible. 
Our method generalizes to finding a
shortest triangle in a weighted disk graph in $O(n \log n)$ expected time.
Moreover, 
we can
compute the unweighted and weighted girth
in a disk graph in $O(n \log n)$ time,
with a deterministic algorithm for the unweighted
case and a randomized algorithm for the weighted
case. The latter result requires a 
method to find a shortest
cycle that contains a given vertex.
Finally, we provide an algorithm to detect a directed triangle
in a transmission graph in $O(n \log n)$ expected time.
For this, we study the geometric properties
of such triangles in more detail, and we develop
several new techniques for batched range searching that
might be of independent interest, using linearized
quadtrees and three-dimensional polytopes to 
test for containment in the union of planar disks. As before, this
algorithm extends to the weighted version.
We will assume \emph{general position},
meaning that all edge lengths (and more generally 
shortest path distances) are pairwise distinct, that 
no site lies on a disk boundary, and that 
all radii are pairwise distinct.

\section{Finding a (Shortest) Triangle in a Disk Graph}
\label{sec:disk_triangle}

We would
like to decide if a given disk graph contains a 
triangle.
If so, we would also like to 
find a triangle of minimum Euclidean perimeter.

\subsection{The Unweighted Case}
\label{sec:disk_triangle_unweighted}
The following property
of disk graphs, due to Evans et al.~\cite{EvansGaLoPo16}, 
is the key to our algorithm.
For completeness, we include a proof.

\begin{lemma}
\label{lem:triangle_planar}
Let $\D(S)$ be a disk graph that is
not plane, i.e., the embedding that 
represents each edge by 
a line segment between its endpoints has two
segments that cross in their relative interiors.  
Then, there are three sites whose associated disks
 intersect in a common point.
\end{lemma}

\begin{figure}
\begin{center}
\includegraphics{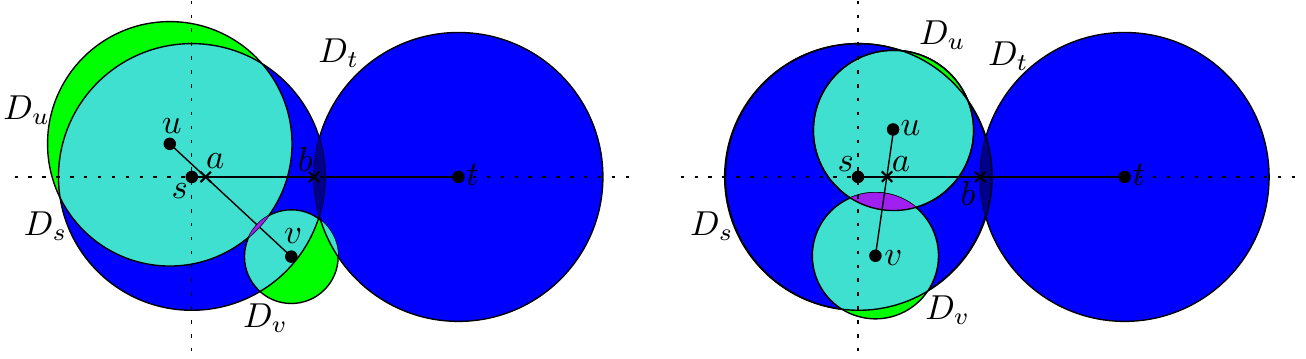}
\end{center}
\caption{If $\D(S)$ is not plane,
then three disks intersect in a common
point. We distinguish two cases, depending on 
whether $u$ lies in the northwest or in the 
northeast quadrant.}
\label{fig:triangle_planar}
\end{figure}
\begin{proof}
Suppose the segments $st$ and
$uv$ intersect in a point $a$. The
sites $s$, $t$, $u$, and $v$ are
pairwise distinct, and without loss
of generality, we assume that:
(i) $a \in D_s \cap D_u$;
(ii) $r_u \leq r_s$;
(iii) the point $s$ lies in the origin, the edge $st$ lies on the 
$x$-axis, with $t$ in the positive halfplane; and
(iv) the site $u$ lies above the $x$-axis,
the site  $v$ lies below the $x$-axis;
see Figure~\ref{fig:triangle_planar}.

If $a \in D_t$, then $D_s \cap D_t \cap D_u \neq \emptyset$,
and we are done. Thus, suppose that $a \not\in D_t$, and 
let $b$ be the first point on $st$ in $D_t$.
If $b \in D_u$, then $D_s \cap D_t \cap D_u \neq \emptyset$,
and we are done. Thus, suppose that $b \not\in D_u$.
If $u$ lies in the northwest quadrant, 
then $v$ must be in the southeast quadrant. Furthermore, 
since $r_u \leq r_s$, it follows that in the southeast 
quadrant, $D_u$ is completely contained in $D_s$, so 
the first point on the segment $av$ that is in $D_v$ must 
also be in $D_s$ and $D_u$. Thus,
$D_s \cap D_u \cap D_v \neq \emptyset$, and we are done.
If $u$ lies in the northeast quadrant, since $r_u \leq r_s$ 
and since $b \not\in D_u$, it follows that 
below the $x$-axis, we have $D_u \subseteq D_s$, and 
the first point on the segment $av$ that is in $D_v$ must 
also be in $D_s$ and $D_u$, i.e.
$D_s \cap D_u \cap D_v \neq \emptyset$.
\end{proof}

If $\D(S)$ is not plane, 
it contains a triangle by
Lemma~\ref{lem:triangle_planar}.
If $\D(S)$ is plane, we can construct it 
explicitly 
and then search for a triangle in
$O(n)$ time~\cite{ItaiRo78,PapadimitriouYa81}.  
To check whether $\D(S)$ is plane, 
we begin an explicit construction of $\D(S)$ and 
abort if we discover too many edges.

\begin{theorem}
\label{thm:unweightedtriangle}
 Let $\D(S)$ be a disk graph on $n$ sites.
 We can find a triangle in  
 $\D(S)$ in $O(n \log n)$ worst-case time, if it
 exists.
\end{theorem}

\begin{proof}
For each \(s\in S\) we split $\partial D_s$ into two $x$-monotone
curves, namely the upper and the lower arc from the
leftmost point of $D_s$ to the rightmost point.
We use the Bentley-Ottmann sweepline algorithm to
find the intersections between these boundary 
arcs. The intersections are reported one 
by one, and the total time to
find the first $m$ intersections is
$O(n\log n + m\log n)$~\cite[Theorem~2.4]{dBCvKO}.\footnote{The 
algorithm is presented for line segments, but it 
extends easily to continuous $x$-monotone curves.}
If the sweepline algorithm
reports more than
$6n - 12$ intersection points, we can be sure
that $\D(S)$ is not plane, because an edge of $\D(S)$ 
corresponds to at most two intersections. 
Then, $\D(S)$ contains a triangle
by Lemma~\ref{lem:triangle_planar}, 
and we can find it
in $O(n \log n)$ time with 
another plane sweep that gives an intersection
between the edges we have generated so far.

If there are at most $6n - 12$ intersections, 
we use another plane sweep to find the 
vertical decomposition of the arrangement of the
disks $D_s$, $s \in S$, in $O(n \log n)$ time. We use the 
vertical decomposition to
construct the remaining edges of
$\D(S)$ that are due 
to a disk being completely contained in another
disk. For this, we walk through the pseudo-trapezoids of
the vertical decomposition, keeping track of the disks
that contain the current pseudo-trapezoid. When we enter
a disk for the first time, we generate edges between
this disk and the disks containing the current pseudo-trapezoid. 
If it turns out that $\D(S)$ has more than $3n - 6$
edges, we abort the generation of the edges, since 
then $\D(S)$ is not plane and contains a 
triangle that can be found in $O(n \log n)$ time 
with a plane sweep over the edges generated so far.
If $\D(S)$ has at most $3n - 6$ edges, 
we obtain an explicit representation of $\D(S)$.
We check if this representation is plane
in $O(n \log n)$ time with a plane sweep, returning
a triangle if this is not the case.
Finally, if $\D(S)$ is plane, we check for a 
triangle in $O(n)$ time~\cite{ItaiRo78,PapadimitriouYa81}.  
\end{proof}

\subsection{The Weighted Case}
\label{sec:disk_triangle_weighted}

Suppose 
the edges in $\D(S)$ are weighted by their
Euclidean lengths. We would like to find
a triangle 
of minimum perimeter, 
i.e., of 
minimum  total 
edge length.
For this, we solve
the decision problem: given 
$W > 0$, does $\D(S)$ contain a triangle with
perimeter at most $W$? Once a decision algorithm
is available, 
the optimization problem  can be solved with
Chan's  randomized geometric
optimization framework~\cite{Chan99}. 

To decide if $\D(S)$ contains a
triangle with perimeter at most $W$, we 
use a grid with diameter
$W/3$.
We look for triangles whose vertices
lie in a single grid cell,
using the algorithm from
Section~\ref{sec:disk_triangle_unweighted}.
If no cell contains
such a triangle, then $\D(S)$
will be sparse and we will need to check only
$O(n)$ further triples. Details follow.

Set $\ell = W/(3\sqrt{2})$.
Let $G_1$ be the grid whose cells are
pairwise disjoint, axis-parallel squares
with side length $\ell$, 
aligned
so that the origin $(0,0)$ is a vertex
of $G_1$.  The cells of $G_1$
have diameter $\sqrt{2} \cdot \ell = W/3$, so 
any triangle whose vertices lie in a single 
cell has perimeter at most $W$.
We make three additional copies
$G_2$, $G_3$, $G_4$ of $G_1$, and we shift
them by $\ell/2$ 
in the $x$-direction, in the $y$-direction, and in both
the $x$- and $y$-directions, respectively.
In other words, $G_2$ has $(\ell/2, 0)$ as a vertex,
$G_3$ has $(0, \ell/2)$ as a vertex, and
$G_4$ has $(\ell/2,\ell/2)$ as a vertex,
see Figure~\ref{fig:grids}. This ensures that if 
all edges in a triangle are ``short'', the triangle lies in a single 
grid cell.

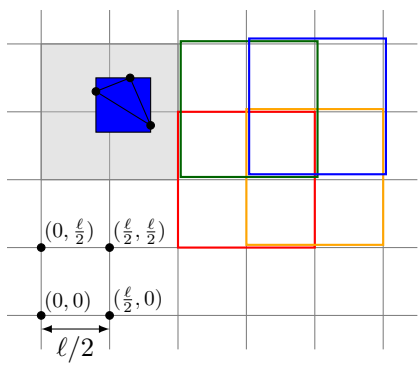
\begin{figure}
\begin{center}
\begin{tikzpicture}[>=latex, scale=0.9]
\foreach \i in {-2,...,3}
{
	\draw[gray] (\i,4.5) -- (\i,-0.5);
}
\foreach \i in {0,...,4}
{
	\draw[gray] (3.5,\i) -- (-2.5, \i);
}

\node[vertex, label={[inner sep=0pt, scale=0.8]above right:\((0,0)\)}] at (-2,0){};
\node[vertex, label={[inner sep=0pt, scale=0.8]above right:\((\frac{\ell}{2},0)\)}] at (-1,0){};
\node[vertex, label={[inner sep=0pt, scale=0.8]above right:\((0,\frac{\ell}{2})\)}] at (-2,1){};
\node[vertex, label={[inner sep=0pt, scale=0.8]above right:\((\frac{\ell}{2},
    \frac{\ell}{2})\)}] at (-1,1){};

\draw[<->] (-2,-0.2) --node[label={[below=0.1cm]\(\ell/2\)}]{} (-1,-0.2);

\path[fill=lightgray, draw=gray, fill opacity=0.4](-2,2) -- (0,2) -- (0,4) -- (-2,4) --cycle;
\coordinate(c1) at (-0.4, 2.8);
\coordinate(c2) at (-1.2,3.3){};
\coordinate(c3) at (-0.7,3.5){};
\draw[fill=ourblue] (-1.2, 3.5) -- (-1.2,2.7) -- (-0.4,2.7) -- (-0.4,3.5) -- cycle;
\node[vertex](v1) at (c1){};
\node[vertex](v2) at (c2){};
\node[vertex](v3) at (c3){};
\draw(v1) -- (v2) -- (v3) -- (v1);

\coordinate(g11) at (0,1);
\coordinate(g12) at (2,1);
\coordinate(g13) at (2,3);
\coordinate(g14) at (0,3);
\draw[ourred, thick] (g11) -- (g12) -- (g13) -- (g14) -- cycle;

\coordinate(g21) at (1,1.04);
\coordinate(g22) at (3,1.04);
\coordinate(g23) at (3,3.04);
\coordinate(g24) at (1,3.04);

\draw[ouryellow, thick] (g21) -- (g22) -- (g23) -- (g24) -- cycle;

\coordinate(g31) at (0.04,2.04);
\coordinate(g32) at (2.04,2.04);
\coordinate(g33) at (2.04,4.04);
\coordinate(g34) at (0.04,4.04);

\draw[ourdarkgreen, thick] (g31) -- (g32) -- (g33) -- (g34) -- cycle;

\coordinate(g41) at (1.04,2.08);
\coordinate(g42) at (3.04,2.08);
\coordinate(g43) at (3.04,4.08);
\coordinate(g44) at (1.04,4.08);

\draw[ourblue, thick] (g41) -- (g42) -- (g43) -- (g44) -- cycle;
\end{tikzpicture}
\end{center}
\caption{The four shifted grids, with a cell
from each grid shown in red, orange, green, and
blue, respectively. Every square
with side length at most $\ell/2$ is wholly contained
in a single grid cell.}
\label{fig:grids}
\end{figure}

\begin{lemma}\label{lem:trianglegridcell}
Let $\Delta$ be a triangle formed by three
vertices $a,b,c \in \R^2$ such that
each edge of $\Delta$ has length at most
$\ell/2$.  There is a cell
$\sigma \in \bigcup_{i = 1}^4 G_i$ with
$a, b, c \in \sigma$.
\end{lemma}

\begin{proof}
We can enclose $\Delta$ with
a square of side length $\ell/2$. 
This square must be completely contained in
a cell of one of the four grids,
see Figure~\ref{fig:grids}.
\end{proof}

We go
through all nonempty grid cells
$\sigma \in \bigcup_{i=1}^4 G_i$, and we search
for a triangle in the disk graph $\D(S \cap \sigma)$
induced by the sites in $\sigma$, with
Theorem~\ref{thm:unweightedtriangle}.  
Since each site lies in $O(1)$ 
grid cells, and since we can compute the grid
cells for a given site in $O(1)$ time (using the floor function),
the total running time is $O(n \log n)$.
If a triangle 
is found, we return
YES, since the
cells have diameter $W/3$ and thus such a triangle
has perimeter at most $W$. If no triangle is found,
Lemma~\ref{lem:trianglegridcell} implies
that any triangle in $\D(S)$ has one side of length 
more than $\ell/2$ and hence at least one vertex with 
associated radius at least $\ell/4$.
We call a site $s \in S$ \emph{large} if $r_s > \ell/4$.
A simple volume argument bounds the number of large sites
in a grid cell.
\begin{lemma}
\label{lem:volumearg}
Let $\sigma \in  \bigcup_{i = 1}^4 G_i$
be a nonempty grid cell, and suppose
that $D(S \cap \sigma)$ does not contain a triangle.
Then $\sigma$ contains at most $18$ large sites.
\end{lemma}

\begin{proof}
Suppose $\sigma$ contains at least
$19$ large sites.
We cover $\sigma$ with $3 \times 3$
congruent squares of side length
$\ell/3$. Then,
at least
one square $\tau$ contains at least
$\lceil 19/9 \rceil = 3$
large sites. The associated
radius of a large site is more than
$\ell/4$ and each square has diameter
$(\sqrt{2}/3)\ell < \ell/2$, 
so the large sites in $\tau$ form a
triangle in $D(S \cap \sigma)$, a contradiction.
\end{proof}

Let $\sigma \in G_i$, $i \in \{1,\dots, 4\}$, be a grid
cell.  The
\emph{neighborhood} $N(\sigma)$ of $\sigma$
is the $5 \times 5$ block of cells in $G_i$
centered at $\sigma$. Since
the diameter of a grid cell is $W/3$, any two sites
$u, v \in S$ that form a triangle of perimeter at most $W$ 
with a site $s \in S \cap \sigma$ must
be in $N(\sigma)$.
Let $S_\ell \subseteq S$ denote the large sites.
At this stage, we know that any triangle in $\D(S)$ has at 
least one vertex in $S_\ell$.
By Lemma~\ref{lem:volumearg}, for any 
$\sigma \in  \bigcup_{i = 1}^4 G_i$,
we have
$|\cup_{\tau \in N(\sigma)} \tau \cap S_\ell| = O(1)$.
Thus, to detect a triangle of perimeter
at most $W$ with at least two large vertices, we
proceed as follows:
for each non-empty cell $\sigma \in G_i$, iterate over all 
large sites $s$ in $\sigma$, over all large sites  $t$ in $N(\sigma)$, 
and over all (not necessary large) sites $u$ in  $N(\sigma)$.
Check whether $stu$ is a triangle of perimeter at most $W$.
If so, return YES.
Since the sites in 
each grid cell are examined $O(1)$ times for $O(1)$ pairs of 
large sites, 
the total time is $O(n)$.

It remains to detect triangles of perimeter at most $W$ with 
exactly one large vertex.
We iterate over all grid cells
$\sigma \in \bigcup_{i=1}^4 G_i$, and we compute 
$\D(S \cap \sigma)$. Since $\D(S \cap \sigma)$ 
contains no triangle, Lemma~\ref{lem:triangle_planar} 
shows that $\D(S \cap \sigma)$ is plane, has 
$O(|S \cap \sigma|)$ edges and can be
constructed in time $O(|S \cap \sigma| \log |S \cap \sigma|)$.
For every edge $st \in \D(S \cap \sigma)$ with both endpoints 
in $S \setminus S_\ell$, we iterate over all large sites $u$ 
in $N(\sigma)$ and we test whether  $stu$ makes a triangle in 
$\D(S)$ with perimeter at most $W$. If so, we return YES. 
By Lemma~\ref{lem:volumearg}, this takes $O(|S \cap \sigma|)$ 
time, so 
the total running time is $O(n \log n)$.
If there is a triangle of perimeter at most $W$ 
with exactly one vertex in $S_\ell$, the edge with both endpoints in 
$S \setminus S_\ell$
has length at most $\ell/2$ and thus must lie in a single 
grid cell $\sigma \in \bigcup_{i=1}^4 G_i$.
To summarize:
\begin{lemma}
\label{lem:decision}
Let $\D(S)$ be a disk graph on $n$
sites, and let $W > 0$. We can decide
in $O(n \log n)$ worst-case time whether
$\D(S)$ contains a triangle of perimeter at
most $W$.
\end{lemma}

We employ the following
general lemma due to Chan~\cite{Chan99}.
Let $\Pi$ be a \emph{problem space}, and 
for a problem $P \in  \Pi$, let 
$w(P) \in \R$ be its \emph{optimum} and
$|P| \in \N$ be its \emph{size}.

\begin{lemma}[Lemma~2.1 in~\cite{Chan99}]
\label{lem:chan}
Let $\alpha < 1$, $\eps > 0$, and $r \in \N$ be
constants, and let $\delta(\cdot)$
be a function such that $\delta(n)/n^\eps$ is
monotone increasing in $n$.
Given any optimization problem $P \in \Pi$ with
optimum $w(P)$, suppose that within
time $\delta(|P|)$,
(i) we can decide whether $w(P) < t$, for
any given $t \in \R$, and (ii)
we can construct $r$ subproblems
$P_1, \dots, P_r$, each of size at
most $\lceil \alpha|P|\rceil$, so that
\[
w(P) = \min\{w(P_1), \dots, w(P_r)\}.
\]
Then we can compute $w(P)$ in total expected time $O(\delta(|P|))$.
\end{lemma}

Now the following main theorem of this section is immediate.

\begin{theorem}
\label{thm:shortesttriangle}
Let $\D(S)$ be a weighted disk graph on $n$ sites.
We can compute a shortest triangle in $\D(S)$
in $O(n \log n)$ expected time, if one exists.
\end{theorem}
\begin{proof}
We apply Lemma~\ref{lem:chan}.
For Condition~(i), we use Lemma~\ref{lem:decision}.
For Condition~(ii), we construct
four subsets $S_0, \dots, S_3$ of $S$
as follows: enumerate the sites in
$S$ as $S = \{s_1, \dots, s_n\}$, and
put the site $s_i$ into
all sets $S_j$ with $i \not\equiv j \pmod 4$.
Then, for any three sites $a,b,c \in S$, there
is at least one subset $S_j$ with
$a,b,c \in S_j$.
Now, Lemma~\ref{lem:chan} with
$\alpha = 3/4$, $\eps = 1$, $r = 4$, and
$\delta = O(n \log n)$ implies the theorem.
\end{proof}

\section{Computing the Girth of a Disk Graph}
\label{sec:disk_girth}

We extend the results from 
Section~\ref{sec:disk_triangle} to the girth.
The unweighted case is 
easy:
if $\D(S)$ is not plane, 
the girth is $3$, 
by Lemma~\ref{lem:triangle_planar}.
If $\D(S)$ is plane, we 
use the algorithm for planar graphs~\cite{ChangLu13}.
The weighted case is harder. If 
$\D(S)$ is plane, we use the algorithm
for planar graphs~\cite{lacki_min-cuts_2011}. 
If not, 
Theorem~\ref{thm:shortesttriangle} gives 
a shortest triangle $\Delta$ in $\D(S)$. 
However, there could be
cycles with at least four edges that are 
shorter than $\Delta$. 
To address this, we use $\Delta$ to split 
$\D(S)$ into sparse pieces where a shortest cycle 
can be found efficiently.

\subsection{The Unweighted Case}
\label{sec:disk_girth_unweighted}

Chang and Lu~\cite[Theorem~1.1]{ChangLu13} showed how 
to find
the girth of an unweighted planar graph 
with $n$ vertices in $O(n)$ time.
Hence, we obtain a simple extension of 
Theorem~\ref{thm:unweightedtriangle}.
\begin{theorem}
\label{thm:disk_unweighted_girth}
Let $\D(S)$ be a disk graph for a set $S$ of $n$ sites.
We can compute the unweighted girth of 
$\D(S)$ in $O(n \log n)$ worst-case time.
\end{theorem}

\begin{proof}
We proceed as in 
Theorem~\ref{thm:unweightedtriangle}.
If $\D(S)$ is not plane, the girth is $3$.
If $\D(S)$ is plane, 
we apply the algorithm of Chang and Lu~\cite[Theorem~1.1]{ChangLu13}
to an explicit representation of $\D(S)$.
\end{proof}

\subsection{The Weighted Case}

We describe how to find 
the shortest cycle through 
a given vertex in a 
weighted graph with certain properties.
 This is then used 
to compute the weighted girth of a disk 
graph.

Let $G$ be a graph 
with nonnegative edge weights so that all 
shortest paths and cycles in $G$ have 
pairwise distinct lengths and so that 
for all edges $uv$, the shortest path 
from $u$ to $v$ is the edge $uv$. 
We present a deterministic 
algorithm that, given $G$ and a vertex $s$, 
computes the shortest cycle in $G$ containing 
$s$, if it exists.\footnote{Even though this seems to be 
  a simple fact, we could not locate a previous reference for 
  this.} A simple randomized algorithm
can also be found in Yuster~\cite[Section~2]{Yuster11}.
The next lemma states a structural property 
of the shortest cycle through $s$.
It resembles
Lemma~1 of Roditty and V.~Williams~\cite{RodittyVa11} that
deals with an overall shortest cycle in $G$.

\begin{lemma}
\label{lem:cycle_struct}
The shortest cycle in $G$ that contains $s$ 
consists of two paths in the shortest path tree 
of $s$, and one additional edge.
\end{lemma}

\begin{proof}
Let $C  = v_0, v_1, v_2, \dots, v_{\ell-1}, v_\ell$ 
be the shortest cycle in $G$ containing $s$, where 
all vertices $v_i$, $0 \leq i \leq \ell - 1$ are 
pairwise distinct, 
$\ell \geq 3$, and $v_0 = v_\ell = s$.
For $v_i \in C$, let $d_1(v_i)$ be the length of 
the path $s, v_1, \dots, v_i$, and let $d_2(v_i)$ be 
the length of the path $v_i, v_{i+1}, \dots, s$. 
Let $\pi(v_i)$ denote the shortest path from $s$ 
to $v_i$, and let $|v_i v_{i+1}|$ be the length of
the edge $v_iv_{i+1}$.

Suppose that $C$ is not of the desired form. 
Let $v_kv_{k+1}$ be the edge on $C$ with
$d_1(v_k) < |v_kv_{k+1}| + d_2(v_{k+1})$ and
$d_2(v_{k+1}) < d_1(v_{k}) + |v_k v_{k+1}|$. By 
our assumptions on $G$, the edge $v_k v_{k+1}$ 
exists and $k \neq 0, \ell - 1$.
We distinguish two cases, illustrated in \cref{fig:dijkstra}.

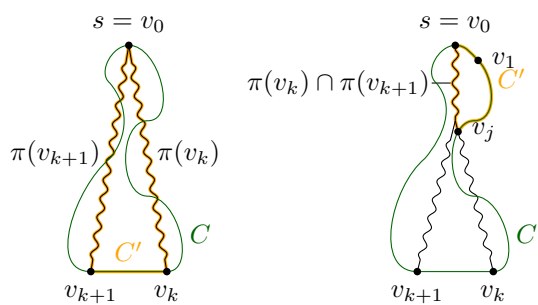
\begin{figure}
\center
\begin{tikzpicture}
\tikzset{dot/.style={fill, circle, inner sep = 1}}
\coordinate(s) at (0,0);
\coordinate(vk) at (0.5,-3);
\coordinate(vk+1) at (-0.5,-3);

\coordinate(m1) at (-0.25,-0.5);
\coordinate(m2) at (0.25,-1);
\coordinate(m3) at (-0.25,-1.5);
\coordinate(m4) at (0.25,-2);
\coordinate(m5) at (-0.25,-2.5);

\draw[very thick, ouryellow, decorate, decoration={snake, amplitude=0.4mm, post length=1.5mm, pre length = 2.3mm, segment length=2.5mm}](s)--(vk);
\draw[very thick, ouryellow,decorate, decoration={snake, amplitude=0.4mm, post length=1.5mm, pre length = 1.5mm, segment length=2.5mm}](vk+1)--(s);

\draw[very thick, ouryellow](vk)--node[above]{$C'$}(vk+1);

\draw[decorate, decoration={snake, amplitude=0.4mm, post length=1.5mm, pre length = 2.3mm, segment length=2.5mm}](s)-- node[right, yshift=0.2em]{$\pi(v_{k})$}(vk);

\draw[decorate, decoration={snake, amplitude=0.4mm, post length=1.5mm, pre length = 1.5mm, segment length=2.5mm}](vk+1)-- node(label)[left,yshift=0.2em]{$\pi(v_{k+1})$}(s);

\draw[ourdarkgreen](vk)--(vk+1);
\draw[ourdarkgreen](s) to[out=0, in =0](m2) to[out=180, in = 180] (m4) to[out=0, in = 0] node[right]{$C$} (vk);
\draw[ourdarkgreen](s) to [out=180, in=140](m1) to [out=-50, in=50](m3) to [out=-140, in =180] (vk+1);

\node[dot, label={$s=v_0$}] at (s) {};
\node[dot, label=below:$v_k$] at (vk) {};
\node[dot, label=below:$v_{k+1}$] at (vk+1) {};

\end{tikzpicture}
\begin{tikzpicture}
\tikzset{dot/.style={fill, circle, inner sep = 1}}
\coordinate(s) at (0,0);
\coordinate(v1)at (0.3,-0.2);
\coordinate(vk) at (0.5,-3);
\coordinate(vk+1) at (-0.5,-3);
\coordinate(split) at (0,-1);

\coordinate(m1) at (-0.25,-0.5);
\coordinate(m2) at (0.25,-1);
\coordinate(m3) at (-0.4,-1.5);
\coordinate(m4) at (0.25,-2);
\coordinate(m5) at (-0.25,-2.5);

\draw[decorate, decoration={snake, amplitude=0.4mm, post length=0mm, pre length = 0.5mm, segment length=2.5mm}, very thick, ouryellow](s)-- (split);

\draw[decorate, decoration={snake, amplitude=0.4mm, post length=1.5mm, pre length = 2.3mm, segment length=2.5mm}, name path=Dk](split)--(vk);
\draw[decorate, decoration={snake, amplitude=0.4mm, post length=1.5mm, pre length = 1.5mm, segment length=2.5mm}](vk+1)--(split);

\draw[decorate, decoration={snake, amplitude=0.4mm, post length=0mm, pre length = 0.5mm, segment length=2.5mm}](s)-- node(label)[left, xshift=-0.6em]{$\pi(v_k)\cap \pi(v_{k+1})$}(split);
\draw ($(label.east)-(0.1,0)$) -- ($(label.east)+(0.15,0)$);


\draw[ourdarkgreen](vk)--(vk+1);
\draw[ourdarkgreen, name path=Ck](s) to [out=0, in=135] (v1) to[out=-45, in =0](m2) to[out=180, in = 180] (m4) to[out=0, in = 0] node[right]{$C$} (vk);
\draw[ourdarkgreen](s) to [out=180, in=140](m1) to [out=-50, in=50](m3) to [out=-140, in =180] (vk+1);
\node[dot, label={$s=v_0$}] at (s) {};
\node[dot, label=below:$v_k$] at (vk) {};
\node[dot, label=below:$v_{k+1}$] at (vk+1) {};
\node[dot, label=right:{$v_1$}] at (v1){};

\path[name intersections={of=Dk and Ck,by=vj}];
\node[dot, label=right:{$v_j$}] at (vj){};

\node[above right of = vj, ouryellow]{$C'$};
\begin{scope}[on background layer]
\draw[very thick, ouryellow](split)--(vj);
\draw[very thick, ouryellow](s) to [out=0, in=135] (v1) to [out=-45, in=0](m2)   to[out=180, in=55] (vj);
\end{scope}
\end{tikzpicture}
\caption{The two cases for $\pi(v_k)\cap \pi(v_{k+1})$. 
On the left the paths are disjoint, on the right 
the shortest path share a prefix.}
\label{fig:dijkstra}
\end{figure}

First, suppose that 
$\pi(v_k) \cap \pi(v_{k+1}) = \{s\}$.
Consider the cycle $C'$ given by $\pi(v_k)$, the 
edge $v_k v_{k+1}$, and $\pi(v_{k+1})$. Since 
$s \neq v_k, v_{k+1}$ and since the edge 
$v_{k}v_{k+1}$ does not appear on $\pi(v_k)$
nor on $\pi(v_{k+1})$, it follows that $C'$ is a 
proper cycle.  Furthermore, by assumption, $C'$ 
is strictly shorter than $C$, because 
$\pi(v_k)$ is shorter than $d_1(v_k)$ or $\pi(v_{k+1})$ 
is shorter than $d_2(v_{k+1})$.
This contradicts our assumption on $C$.

Second, suppose that
$|\pi(v_k) \cap \pi(v_{k+1})| \geq 2$.
Since $\pi(v_k)$ and $\pi(v_{k+1})$ are shortest
paths, their intersection is a prefix of each path.
By the assumption\footnote{Namely, that for all 
edges $uv$, the shortest path from $u$ to $v$ is the edge 
$uv$.}
on $G$ 
at least one of $v_1, v_{\ell-1}$ is not in
$\pi(v_k) \cup \pi(v_{k+1})$.
Without loss of generality, this vertex is $v_1$.
Let $j \geq 1$ be the smallest index such that 
$v_j \in \pi(v_k) \cup \pi(v_{k+1})$. We have 
$j \in \{2, \dots, k\}$.  Consider the cycle $C'$ 
that starts at $s$, follows $C$ along
$v_1, v_2, \dots$ up to $v_j$, and then returns 
along $\pi(v_k)$ or $\pi(v_{k+1})$ to $s$. By 
construction, $C'$ is a proper cycle.
Furthermore, $C' \neq C$, because even if $j = k$, 
the path $\pi(v_k)$ cannot contain the part of $C$ 
from $v_{k+1}$ to $s$,
due to the choice of $k$. Finally, $C'$
is strictly shorter than $C$, because the second part 
of $C'$ from $v_j$ to $s$ follows a shortest path and 
is thus strictly shorter than
$d_2(v_j)$. Again, $C'$ contradicts our choice of $C$.
\end{proof}

\begin{theorem}
\label{thm:dijkstra_cycle}
Let $G = (V, E)$ be a weighted graph with $n$ vertices
and $m$ edges that has the properties
given at the beginning of this section. Let $s \in V$.
We can compute the shortest cycle in $G$ that contains $s$ 
in $O(n \log n + m)$ time, if it exists.
\end{theorem}

\begin{proof}
We find the shortest path
tree $T$ for $s$ in $G$, and
we traverse $T$ to find for each vertex $v$ in 
$T \setminus \{s\}$ the second vertex $b[v]$ on the
shortest path from $s$ to $v$ (the vertex following $s$). 
Then, we iterate over all edges in $E$  that are
not in $T$. For each such $e = uv$, 
we check if $b[u] \neq b[v]$. If so, $e$ 
closes a cycle in $T$ that contains $s$. 
We determine the length of this cycle (in $O(1)$ time).
We return the shortest cycle found in this way.

The correctness follows from
Lemma~\ref{lem:cycle_struct}. As for the running time,
it takes  $O(n \log n + m)$ time to find the shortest 
path tree for $s$ with
Dijkstra's algorithm and Fibonacci 
heaps~\cite[Chapter~24.3]{CormenLeRiSt09}.
After that, it takes $O(n)$ time
to compute the nodes $b[v]$, for $v \in T \setminus \{s\}$,
and $O(m)$ time to iterate over the edges not in $T$.
The length of the cycle associated with such an edge $e$ 
can be computed in $O(1)$ time, using the shortest path
distances in $T$ and the length of $e$.
\end{proof}

Let $\D(S)$ be a weighted
disk graph on $n$ sites.  
A careful combination of the tools developed so far
gives an algorithm for the weighted girth of $D(S)$.

\begin{theorem}
\label{thm:disk_girth_weighted}
Given a weighted disk graph $\D(S)$ on 
$n$ sites, we can compute the weighted girth of 
$\D(S)$ in  $O(n\log n)$ expected time.
\end{theorem}

\begin{proof}
We use Theorem~\ref{thm:shortesttriangle} to
find the shortest triangle in $\D(S)$, 
if it exists, in $O(n \log n)$ expected 
time. If $\D(S)$ has no triangle, 
it is plane by Lemma~\ref{lem:triangle_planar}. 
As in the proof of 
Theorem~\ref{thm:unweightedtriangle}, we can  then
explicitly construct $\D(S)$ in $O(n\log n)$ time
with a plane sweep. We 
determine the girth of $\D(S)$ using 
the algorithm of \Lacki and Sankowski~\cite[Section~5]{lacki_min-cuts_2011},
in additional $O(n \log\log n)$ time, and are done.

Now, suppose $\D(S)$ contains a triangle, 
and let $W$ be the length of the shortest
triangle in $\D(S)$, an upper bound for 
the girth of $\D(S)$.
As in Section~\ref{sec:disk_triangle_weighted},
we set $\ell = W/(3\sqrt{2})$, and we let 
let $G$ be the grid of
side length $\ell$ 
that has the origin
$(0,0)$ as a vertex.
We call a
site $s \in S$ 
\emph{large} if $r_s \geq \ell/4$,
and we let $S_\ell \subseteq S$ be the set of large sites.

We need to check whether $\D(S)$ contains
a cycle with more than three vertices
and length less than $W$. If so, we must
find the shortest such cycle. First, we consider 
cycles in the induced subgraph $\D(S \setminus S_\ell)$.
The graph $\D(S \setminus S_\ell)$ has no
triangle,
as such a triangle would have length less than 
$3 \cdot \ell/2 < W$.
Thus, 
by Lemma~\ref{lem:triangle_planar},
$\D(S \setminus S_\ell)$ is plane.
We can directly compute the 
weighted girth of $\D(S \setminus S_\ell)$ in $O(n \log n)$ time
with a plane sweep and 
the algorithm of \Lacki and Sankowski~\cite[Section~5]{lacki_min-cuts_2011}.
Let $\Delta_1$ be the weighted girth of $\D(S \setminus S_\ell)$.

Next, we consider cycles that have at least one large site. 
Let $\sigma$ be a cell of $G$.
The induced subgraph 
$\D(S \cap \sigma)$ has
no triangle, since by the choice of $\ell$,
such a triangle would have length less than $W$.
Thus, Lemma~\ref{lem:volumearg} 
shows $|S_\ell \cap \sigma| = O(1)$.
By the triangle inequality, the maximum distance
between any two sites in a cycle of length less 
than $W$ is less than $W/2$.  Thus, any such cycle 
containing a site $s \in S \cap \sigma$ completely 
lies 
in the $7 \times 7$ neighborhood 
$N(\sigma)$ around $\sigma$.
Since $N(\sigma)$ has $O(1)$ 
cells and since each cell contains $O(1)$ large sites, 
there are $O(1)$ large sites in 
$S_\sigma = S \cap (\bigcup_{\tau \in N(\sigma)} \tau)$.

Now, for each grid cell $\sigma$, we consider all 
large sites $s \in S_\ell \cap \sigma$.
We must find the shortest cycle 
through $s$ in the subgraph $\D(S_\sigma)$ of $\D(S)$ 
in $N(\sigma)$.
Let $n_\sigma = |S_\sigma|$. 
Since the graph induced 
by $S_\sigma \setminus S_\ell$ is plane and since 
$|S_\sigma \cap S_\ell| = O(1)$, the graph $\D(S_\sigma)$ has 
$O(n_\sigma)$ edges.  Hence, we can construct 
$\D(S_\sigma)$ and apply Theorem~\ref{thm:dijkstra_cycle} 
to compute the shortest cycle in $\D(S_\sigma)$
through $s$ in total time $O(n_\sigma \log n_\sigma)$.
Let $\Delta_2$ be the smallest length of such a 
cycle, over all grid cells $\sigma$ and all large sites 
$s \in S_\ell \cap \sigma$.
Since each small site is involved in $O(1)$
neighborhoods, we get
$\sum_{\sigma \in G} n_\sigma = O(n)$, 
and the overall running time of this step is 
$O(n\log n)$.  Finally, we return 
$\min \{W, \Delta_1, \Delta_2\}$. If we also
want the shortest cycle itself, we simply
maintain appropriate pointers in the algorithm.
The total expected running time  is 
$O(n \log n)$.
\end{proof}

\section{Finding a Triangle in a Transmission Graph}\label{sec:trianglesdir}

Given a transmission graph $T(S)$ 
on $n$ sites, we 
want to decide if $T(S)$ contains 
a directed triangle.
We first describe an inefficient algorithm
for this problem, and then we will explain
how to implement it in $O(n \log n)$ expected time.

The algorithm iterates over 
each directed edge $e = st$ with 
$r_t  \geq r_s$, and it performs two tests:
first, for each directed edge $tu$ 
with $r_u \geq r_t/2$, it checks if
$us$ is an edge in $T(S)$, i.e., if
$s \in D_u$. If so, the algorithm reports the triangle
$stu$. Second, the algorithm tests if there is
a site $u$ such that $r_u \in [r_s, r_t/2)$
and such that $us$ is an edge in $T(S)$, i.e.,
such that $s \in D_u$.
If such a $u$ exists, it reports
the triangle $stu$. 
If both tests fail for each edge $e$,
the algorithm
reports that $T(S)$ contains no triangle.
The next lemma shows that the algorithm is correct.

\begin{lemma}
\label{lem:strategy_correct}
A triple $stu$ reported by the algorithm
is a triangle in $T(S)$.
Furthermore, if $T(S)$ contains a triangle,
the algorithm will find one.
\end{lemma}

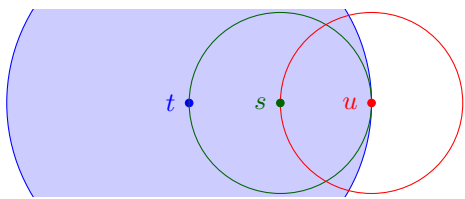
\begin{figure}
\label{fig:requrementecheck}
\center
\begin{tikzpicture}[
site/.style = {draw, circle, inner sep = 1pt, fill},
scale=0.4]
\clip(9.1,-3.1) -- (-6.1,-3.1) -- (-6.1,3.1) -- (9.1,3.1) -- cycle;
\coordinate (t) at (0,0);
\node[ourblue,site, label={[ourblue]left:$t$}] at (t){};
\draw[ourblue, fill opacity=0.2, fill=ourblue] (t) circle (6);
\coordinate (s) at (3,0);
\draw[ourdarkgreen] (s) circle (3);
\coordinate (u) at (6,0);
\node[ourred, site, label={[ourred]left:$u$}] at (u){};
\draw[ourred] (u) circle (3);
\node[ourdarkgreen, site, label={[ourdarkgreen]left:$s$}] at (s){};
\end{tikzpicture}
\caption{We do not need to check $u \in D_t$.}
\end{figure}

\begin{proof}
Let $stu$ be a triple reported by the algorithm.
The algorithm explicitly checks that 
$st$ and $us$ are edges in $T(S)$.
It remains to consider $tu$. 
If $r_u \geq r_t/2$, then $stu$ is reported by
the first test, and the algorithm explicitly checks
that $tu$ is an edge in $T(S)$.
If $r_u < r_t/2$, then $stu$ is reported by
the second test.
We have $r_s < r_t/2$, since $s$ and $t$ are chosen so that
$r_s \leq r_u$. Furthermore, $st$ and $us$ are edges of $T(S)$,
so $t \in D_s$ and $s \in D_u$. Since the second test ensures
that $r_u \leq r_t/2$, it follows from the triangle inequality that
\[
|tu| \leq |ts| + |su| \leq r_s + r_u < r_t/2 + r_t/2 = r_t,
\]
so $u \in D_t$, and $tu$ is an edge in $T(S)$. 
Thus, the reported triple $stu$ is a triangle in $T(S)$.

Now suppose that $T(S)$ contains a 
triangle $stu$, labeled such that 
$r_s \leq \min\{r_t, r_u\}$. If $r_u \geq r_t/2$,
then $stu$ is found by the first test for the 
edge $st$. 
If $r_u < r_t/2$,
we have
$s \in D_u$ and $r_u \in [r_s, r_t/2)$. 
Thus, the second test will be successful
for the edge $st$,
and the algorithm will report a triple $stu'$,
such that 
$s \in D_{u'}$ and $r_{u'} \in [r_s, r_t/2)$
(the site $u'$ might be different from $u$).
The first
part of the proof shows that
$stu'$  is a triangle in $T(S)$.
\end{proof}

There are several challenges for making the 
algorithm efficient.
First of all,
there might be many edges
$st$ with $r_t \geq r_s$.
However, the following
lemma shows that if there are $\omega(n)$ such edges,
the transmission graph $T(S)$ must contain a triangle.

\begin{figure}
\label{fig:subdiv}
\center
\begin{tikzpicture}[>=stealth, scale=0.8]

\foreach \i in {-2,...,3}
{
	\draw[lightgray] (\i,3.5) -- (\i,-0.5);
}
\foreach \i in {0,...,3}
{
	\draw[lightgray] (3.5,\i) -- (-2.5, \i);
}

\node[draw,fill, circle, inner sep = 1, ourblue](s1) at (0.8,1.5){};
\node[draw,fill, circle, inner sep = 1, ourdarkgreen](s2) at (.1,1.3){};
\node[draw,fill, circle, inner sep = 1, ourred](s3) at (.3,1.9){};
\coordinate(ongreen) at ($(s1)-(0,1.5)$);
\draw[ourblue, <->] (s1) -- node[left, scale=0.8]{$\frac{r}{4}$} (ongreen);
\draw[ourblue] (s1) circle (1.5);
\draw[ourdarkgreen] (s2) circle (1.5);
\draw[ourred] (s3) circle (1.5);
\draw[<->](s1)--(s2);
\draw[<->](s2)--(s3);
\draw[<->](s3)--(s1);
\draw[<->](2,0.4)  --node[above, scale=0.8] {$\frac{r}{6}$} (3,0.4);
\end{tikzpicture}
\caption{Three disks with radius at least $r/4$ in the same 
grid cell form a clique.}
\end{figure}
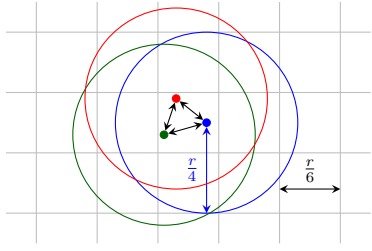

\begin{lemma}\label{lem:disks_close_together}
There is an absolute constant $\alpha$ so that 
for any $r > 0$,
if there is an $r \times r$ square $\sigma$ 
that contains more than $\alpha$ sites 
$s \in S$ with $r_s \geq r/4$, then $T(S)$ 
has a directed triangle.
\end{lemma}

\begin{proof}
We cover $\sigma$ with a $6 \times 6$ grid of 
side length $r/6$; see \Cref{fig:subdiv}. 
There are $36$ grid cells.
For every $s \in S \cap \sigma$ with
$r_s \geq r/4$, the disk $D_s$ completely 
covers the grid cell containing $s$.
If $\sigma$ contains more than $\alpha = 72$ sites $s$ with 
$r_s \geq r/4$, then
one grid cell contains at least three 
such sites.
These sites form a directed triangle in $T(S)$. 
\end{proof}

Thus, to implement the algorithm, we must solve
two range searching problems.
\begin{description}
\item[(\typeROne)] EITHER determine that for every site $s \in S$,
  there are at most $\alpha$ outgoing edges 
  $st$ with $r_t \geq r_s/2$ 
  and report all
  these edges; OR find a square $\sigma$ of side length
  $r > 0$ that contains more than $\alpha$ sites
  $s \in S$ with $r_s \geq r/4$.
\item[(\typeRTwo)] Given $O(n)$ query triples 
  $(s, r_1, r_2)$ with $s \in S$ and $0 < r_1 < r_2$, 
  find a site $u \in S$ such that
  there is a query triple $(s, r_1, r_2)$ 
  with $u \neq s$,
  $r_u \in [r_1, r_2)$, and $s \in D_u$; or report that
  no such site exists. 
\end{description}
The query (\typeROne) indeed always has a valid outcome: suppose 
there is a site $s \in S$ with more than $\alpha$ outgoing edges
$st$ with $r_t \geq r_s/2$.
Then, all the endpoints $t$ lie in $D_s$, so 
the square $\sigma$ centered at $s$ with side
length $r = 2r_s$ contains more than $\alpha$ sites with 
associated radius at least $r/4$.
The next theorem shows that 
we can detect a triangle in $T(S)$ with 
linear overhead in addition to the
time needed for answering (\typeROne) and (\typeRTwo).

\begin{theorem}
\label{thm:userange}
If (\typeROne) and 
(\typeRTwo) can be solved in 
time $R(n)$ for input size $n$,
we can find a directed triangle in a transmission
graph $T(S)$ on $n$ sites in time
$R(n) + O(n)$, if it exists.
\end{theorem}
\begin{proof}
First, we perform a range query (\typeROne). If 
it reports a square $\sigma$ of side length $r$ 
that contains more than $\alpha$ sites $s \in S$
with $r_s \geq r/4$, we scan 
$S$ to find a set $S'$ of $\alpha + 1$ such sites. 
By Lemma~\ref{lem:disks_close_together},
$T(S')$ contains a triangle, and
we find it in $O(1)$ time by testing all 
triples in $S'$.

Otherwise, (\typeROne) reports the set 
$E'$ of all edges $st$ in $T(S)$ with $r_t \geq r_s/2$, 
where $|E'| = O(n)$. 
We go through all edges $e = st$ in $E'$
with $r_t \geq r_s$, and we check if there
is an edge $tu$ in $E'$
such that $us$ is an edge in $T(S)$, i.e., such that
$s \in D_u$. If so, we
report the triangle $stu$. This takes care of
the first test in the algorithm, and we check
only $O(n)$ triples, because for each site in $S$,
there are at most $\alpha$ outgoing edges 
in $E'$.
If we have not been successful,
we again go trough all edges 
$e = st$ in $E'$, and if $r_t > 2r_s$,
we create the triple $(s, r_s, r_t/2)$.
We perform a range query
(\typeRTwo) on the resulting set of  $O(n)$ triples.
If (\typeRTwo) finds a site $u \in S$ such that for a query triple
$(s, r_s, r_t/2)$, we have  $u \neq s$,
$r_u \in [r_s, r_t/2)$, and $s \in D_u$, we report the
triangle $stu$. Otherwise, we report that $T(S)$ does not
contain a triangle.

By Lemma~\ref{lem:strategy_correct},
we correctly report a triangle in $T(S)$, if 
it exists.
The time for the additional steps 
is $O(n)$, so the total running time 
is $R(n) + O(n)$.
\end{proof}

Using existing methods~\cite{WillardLu85}, it is easy to solve 
(\typeROne) and (\typeRTwo) in $O(n \log^2 n)$ time. However, 
a better solution is possible.
In the next section, we will implement  (\typeROne) and
(\typeRTwo) in $O(n \log n)$ expected time.

\begin{theorem}\label{thm:unweightedtriangledir}
\label{thm:transmission_triangle_unweighted}
Let $T(S)$ be a transmission graph on $n$ sites.
We can find a directed triangle in $T(S)$ in 
expected time
$O(n \log n)$, if it exists.
\end{theorem}

\section{Batched Range Searching}

The range queries must handle subsets of
sites whose associated radii lie in certain intervals:
a query $s$ in (\typeROne)
concerns sites $t \in S$ such that $r_t \geq r_s/2$; and a
query $(s, r_1, r_2)$ in (\typeRTwo) concerns sites 
$t$ such that $r_t \in [r_1, r_2)$.
Using a standard approach~\cite{dBCvKO,WillardLu85}, we 
subdivide each such \emph{query interval} into $O(\log n)$ pieces 
from a set of \emph{canonical intervals}.
For this, we build a balanced binary tree 
$B$ whose
leaves are the sites of $S$, sorted
by increasing 
associated radius.
For each vertex $v \in B$, let 
the
\emph{canonical interval}
$\cD_v$ be the
sorted list of sites in the subtree rooted at 
$v$. There are
$O(n)$ canonical intervals. 

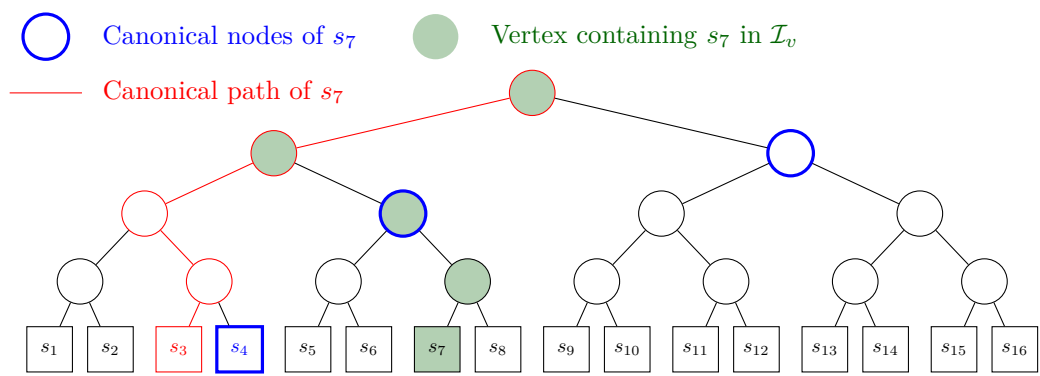
\begin{figure}
\centering
\begin{tikzpicture}[
	inner/.style={draw, circle, minimum size=8mm, black, scale=0.75}, 
	leaf/.style={draw,minimum size=8mm , black, scale=0.75},
	query/.style={fill=ourdarkgreen, text opacity= 1, fill opacity = 0.3},
	rightinner/.style={inner, very thick, ourblue}, 
	every child/.style={black},
	rightleaf/.style={leaf, very thick, ourblue},
	level 1/.style= {sibling distance = 68mm, level distance = 8mm},
	level 2/.style= {sibling distance = 34mm},
	level 3/.style= {sibling distance = 17mm, level distance = 9 mm},
	level 4/.style= {sibling distance = 8mm, level distance = 9mm},
  ]

\begin{scope}[shift={(-7,0)}]
\node (origin) at (0,-1) {};
\node (rightoforigin) at (1.2,-1){};
\node (red) at ($(origin)!0.5!(rightoforigin)$) {};
\draw[red] (origin) -- (rightoforigin);
\node[red, right = 0.5cm of red, anchor = west] (redlabel) {Canonical path of $s_7$};
\node[circle, minimum size=6mm, draw, very thick, ourblue, above = 0.3cm of red] (blue)  {};
\node[ourblue, right =0.3cm of blue, anchor=west] {Canonical nodes of $s_7$};
\node[circle, minimum size=6mm, query, right =4.5cm of blue] (query) {};
\node[ourdarkgreen, right =0.3cm of query, anchor=west]  {Vertex containing $s_7$ in $\cD_v$};
\end{scope}
\begin{scope}[shift={(0,-1)}]
\node[inner, ourred, query]{} 
 child[ourred] {node[inner, ourred, query] {}
  child [ourred]{node[inner, ourred] {}
   child  {node[inner] {}
    child {node[leaf] {$s_1$}}
child {node[leaf] {$s_2$}}
}
   child[ourred] {node[inner, ourred] {}
    child[ourred] {node[leaf, ourred] {$s_3$}}
child {node[rightleaf] {$s_4$}}
}
}
  child {node[rightinner, query] {}
   child {node[inner] {}
    child {node[leaf] {$s_5$}}
child {node[leaf] {$s_6$}}
}
   child {node[inner, query] {}
    child {node[leaf, query] {$s_7$}}
child {node[leaf] {$s_8$}}
}
}
}
 child {node[rightinner] {}
  child {node[inner] {}
   child {node[inner] {}
    child {node[leaf] {$s_9$}}
child {node[leaf] {$s_{10}$}}
}
   child {node[inner] {}
    child {node[leaf] {$s_{11}$}}
child {node[leaf] {$s_{12}$}}
}
}
  child {node[inner] {}
   child {node[inner] {}
    child {node[leaf] {$s_{13}$}}
child {node[leaf] {$s_{14}$}}
}
   child {node[inner] {}
    child {node[leaf] {$s_{15}$}}
child {node[leaf] {$s_{16}$}}
}
}
}
;
\end{scope}
\end{tikzpicture}
\caption{Example is 
for a query of type (\typeROne), assuming that 
\(r_{s_3}< r_{s_7}/2 \leq r_{s_4}\).}
\label{fig:querypath}
\end{figure}

Next, we define \emph{canonical paths}
and \emph{canonical nodes}.
For a radius $r > 0$, the (proper) \emph{predecessor} of $r$ 
is the site $s \in S$ with the largest 
radius $r_s \leq r$ ($r_s < r$).
The (proper) \emph{successor} 
of $r$ is defined analogously.
For a query $s$ in (\typeROne), we consider the path
$\pi$ in $B$ from the root to the leaf with the proper
predecessor $t$ of $r_s/2$.  If $t$ does not
exist (i.e., if $r_t  \geq r_s/2$, for all $t \in S$), 
we let $\pi$ be the
left spine of $B$. We call $\pi$ the \emph{canonical path} 
for $s$. The \emph{canonical nodes} 
for $s$ are the right children of the nodes in $\pi$ that 
are not in $\pi$ themselves, plus possibly the last node 
of $\pi$, if $r_t  \geq r_s/2$, for all $t \in S$, see
Figure~\ref{fig:querypath}.

For a query $(s, r_1, r_2)$ in (\typeRTwo), we consider 
the path $\pi_1$ in $B$ from the root to the leaf
with the proper predecessor $t_1$ of $r_1$ and the path $\pi_2$ in 
$B$ from the root to the leaf for the
 successor $t_2$ of $r_2$. Again, if $t_1$ 
does not exist, we take $\pi_1$ as the left spine of $B$,
and if $t_2$ does not exist,
we take $\pi_2$ as the right spine of $B$.
Then, $\pi_1$ and $\pi_2$ are the \emph{canonical paths}
for $(s, r_1, r_2)$. The \emph{canonical nodes} for $(s, r_1, r_2)$
are defined as follows: for each vertex $v$ in $\pi_1 \setminus \pi_2$, 
we take the right child of $v$ if it is not in $\pi_1$, and
for each $v$ in $\pi_2 \setminus \pi_1$, we take the left child
of $v$ if it is not  in $\pi_1$. Furthermore, we take
the last node of $\pi_1$ if $t_1$ does not exist, and
the last node of $\pi_2$ if $t_2$ does not exist.
A standard argument bounds the number and total size of the 
canonical intervals.

\begin{lemma}
\label{lem:canonical}
The total size of the canonical intervals
is $O(n \log n)$. The tree $B$ and the 
canonical intervals can be built in $O(n \log n)$ time.
For any query $q$ in (\typeROne) or (\typeRTwo), there are
$O(\log n)$ canonical nodes, and they can be found in 
$O(\log n)$ time. 
The canonical intervals for the canonical nodes
of $q$ constitute a partition of the query interval
for $q$.
\end{lemma}
\begin{proof}
Since a site $s \in S$ appears in $O(\log n)$
canonical intervals, the total size of the
canonical intervals is $O(n \log n)$. To construct
$B$, we sort $S$ according to the radii $r_s$,
and we build $B$ on top of the sorted list. To find
the (sorted) canonical intervals, we perform
a postorder traversal of $B$, copying and merging the
child intervals for each parent node.

The bound on the canonical nodes for $q$ follows,
since $B$ has height $O(\log n)$. To find them,
we trace the canonical
paths for $q$ in $B$. The partition
property holds by construction.
\end{proof}

\subsection{Queries of Type (\typeROne)}
\label{sec:rone}

We build a compressed
quadtree on $S$, and we perform the range searches
the compressed quadtree. 
It is possible to compute a compressed 
quadtree for each canonical interval without 
logarithmic overhead. Since Lemma~\ref{lem:disks_close_together} 
gives us plenty of
freedom in choosing the squares for our range queries,
we take squares from the grid that underlies the 
quadtree. This allows us to reduce the range searching problem
to predecessor search in a linear list, a task that
can be accomplished by one top-down traversal of $B$. Details follow.

\subparagraph*{Hierarchical grids, Z-order, compressed 
quadtrees.}
We translate and scale $S$ (and the associated radii), 
so that $S$ lies in the interior of the unit square 
$U = [0,1]^2$ and so that all radii are at most $\sqrt{2}$. 
We define a sequence 
of hierarchical grids that subdivide $U$.
The grid $G_0$ consists of the single cell $U$.
The grid $G_i$, $i \geq 1$, consists of the $2^{2i}$ square 
cells with side length $2^{-i}$ and pairwise
disjoint interiors that cover $U$. The hierarchical
grids induce an infinite four-regular tree 
$\T$: the vertices are the cells of
$\G = \bigcup_{i = 0}^\infty G_i$.
The unit square $U$ is the root, and for $i = 1, \dots$, a 
cell $\sigma$ in $G_i$ is the child of the 
cell in $G_{i-1}$ that contains it.
We make no explicit 
distinction between a 
vertex of $\T$ and its corresponding cell.

\begin{figure}
\center
\begin{tikzpicture}[inner/.style = {draw, circle, minimum size= 4mm,inner sep=0pt}]
\coordinate (topleft) at (0,2);
\coordinate (topright) at (2,2);
\coordinate (bottomleft) at (0,0);
\coordinate (bottomright) at (2,0);
\coordinate (midleft) at (0,1);
\coordinate (midright) at (2,1);
\coordinate (midbottom) at (1,0);
\coordinate (midtop) at (1,2);
\coordinate (center) at (1,1);

\fill[ouryellow] (topleft) -- (midtop) -- (center) -- (midleft) -- (topleft);
\fill[ourblue] (midtop) -- (center) -- (midright) -- (topright);
\fill[ourred] (midleft) -- (center) -- (midbottom) -- (bottomleft);
\fill[ourdarkgreen]  (center) -- (midbottom) -- (bottomright) -- (midright);
\draw (topleft) -- (topright) -- (bottomright) -- (bottomleft) -- (topleft);
\draw (midleft) -- (midright);
\draw (midbottom) -- (midtop);
\path[draw, thick, >=stealth, >->] (0.5,1.5) -> (1.5,1.5) -> (0.5,0.5) -> (1.5,0.5);

\begin{scope}[scale=1.5, shift={(3,1.25)}, level distance = 2em, level 1/.style={sibling distance=1.2em}]

  \node[inner, fill= black]{}
  child {node[inner, fill=ouryellow]{}}
  child {node[inner, fill=ourblue]{}}
  child {node[inner, fill=ourred]{}}
  child {node[inner, fill=ourdarkgreen]{}};
\end{scope}
\begin{scope}[scale=1.5, shift={(5.5,1.25)}, level distance = 2em, level 1/.style={sibling distance=1.2em}]
\tikzstyle{subtree}= [draw,isosceles triangle,shape border rotate=90, anchor = north,isosceles triangle stretches, minimum height=15mm, minimum width=12mm]
\tikzset{decoration={snake,amplitude=.4mm,segment length=1.5mm,
                       post length=0mm,pre length=0mm}}
  \node[inner, fill= black, text=white]{\(\rho\)}
  child[child anchor= apex] {
   node[subtree]{}
   node[inner, fill=ouryellow](T){\(\tau\)}
   node[inner, below right=9.5mm and -1mm of  T](S){\(\sigma\)}
   edge[decorate](T)--(S)
   }
  child{
  node[inner, fill=ourblue]{}
  }
  child[child anchor=apex] {
  node[subtree]{}
  node[inner, fill=ourred](R1){}
  node[inner, below left=7.5mm and -1.3 mm of  R1](R){\(\tilde{\sigma}\)}
  edge[decorate](R1)--(R)
   }
  child {node[inner, fill=ourdarkgreen]{}};
\end{scope}
\end{tikzpicture}
\caption{$Z$-Order. On the very right we have \(\sigma \leq_Z \tau \leq_Z \tilde{\sigma}\).}
\label{fig:zorder}
\end{figure}
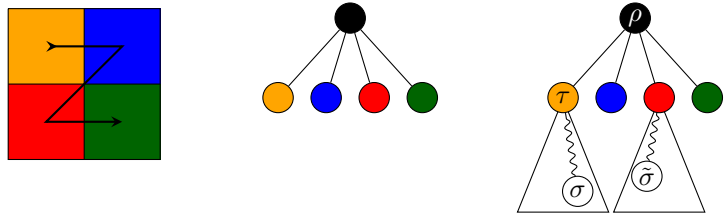

The \emph{$Z$-order} $\leq_Z$ is
a total order on the cells of $\G$; 
see~\cite{BuchinMu11} for more details.
Let $\sigma, \tau \in \G$.
If $\sigma \subseteq \tau$, 
then $\sigma \leq_Z \tau$: and if $\tau\subseteq \sigma$,
then $\tau \leq_Z \sigma$,
If $\sigma$ and $\tau$ are unrelated in $\T$,
let $\rho$ be the lowest common ancestor of 
$\sigma$ and $\tau$ in $\T$, and let $\sigma'$ and $\tau'$
be the children of $\rho$ with $\sigma \subseteq \sigma'$
and $\tau \subseteq \tau'$.
We set $\sigma \leq_Z \tau$ 
if $\sigma'$ is before $\tau'$ in the
order shown in \cref{fig:zorder};
and $\tau \leq_Z \sigma$, otherwise.
The next lemma shows that
given $\sigma, \tau \in \G$, we can decide if
$\sigma \leq_Z \tau$ 
in constant time.

\begin{lemma}[Chapter~2 in Har-Peled~\cite{har-peled_geometric_2008}]
\label{obs:compareconst}
Suppose the floor function 
and the first differing bit in the binary 
representations of two given real numbers can be computed 
in $O(1)$ time.
Then, we can decide in $O(1)$ time for two given
cells $\sigma, \tau \in \G$  whether 
$\sigma \leq_Z \tau$ or $\tau \leq_Z \sigma$.
\end{lemma}

For a site
$s \in S$, let $\sigma_s$ be the largest cell in $\G$ that 
contains only $s$. The \emph{quadtree} for $S$ is 
the smallest connected subtree of $\T$ that contains the root $U$
and all cells $\sigma_s$, for $s \in S$.
The \emph{compressed quadtree} $\C$ for $S$ is obtained from the
quadtree by contracting any maximal
path of vertices with only one child into a single
edge.  Vertices that were at the top of 
such a path are now called \emph{compressed} vertices. 
The compressed quadtree for $S$ has 
$O(n)$ vertices, and it can be constructed in $O(n\log n)$ 
time (see, e.g.,~\cite[Appendix~A]{BuchinLoMoMu11} and 
\cite{har-peled_geometric_2008}).

The \emph{linearized compressed quadtree} $\cL$ for $S$ is
the sorted sequence of cells obtained by listing the nodes of 
$\C$ according to a postorder traversal, were the children of
a node $\sigma \in \C$ are visited according to the $Z$-order
from Figure~\ref{fig:zorder}. The cells
in $\cL$ appear in increasing $Z$-order, 
and range searching for a given cell 
$\sigma \in \G$ reduces to a simple predecessor search
in $\cL$, as is made explicit in the following lemma.

\begin{lemma}\label{lem:quadtreeorder}
Let $\sigma$ be a cell of $\G$, and let $\cL$ be the 
linearized compressed quadtree on $S$. Let 
$\tau = \max_Z \{\rho \in \cL \mid \rho \leq_Z \sigma\}$ be 
the \emph{$Z$-predecessor} of $\sigma$ in $\cL$ ($\tau = \emptyset$,
if the predecessor does not exist).
Then, if $\sigma \cap \tau = \emptyset$, then also
$\sigma \cap S = \emptyset$, and if $\sigma \cap \tau \neq \emptyset$,
then $\sigma \cap S = \tau \cap S$.
\end{lemma}

\begin{proof}
Let $\C$ be the compressed quadtree on $S$, and 
let $\C_\sigma = \{ \tau \in \C \mid \tau \subseteq \sigma\}$ 
be the cells in 
$\C$ that are contained in $\sigma$. If $\C_\sigma$
is non-empty, then $\C_\sigma$ is a connected subtree
of $\C$. Let $\tau$ be the root of this subtree.
Then, $\tau = \max_Z \{\rho \in \C_\sigma\}$, and
$\tau \leq_Z \sigma$. Furthermore, all other cells
in $\C \setminus \C_\sigma$ are either smaller than
all cells in $\C_\sigma$ or larger than $\sigma$.
Thus, $\tau$ is the $Z$-predecessor of $\sigma$  in 
$\cL$, and $\sigma \cap S = \tau \cap S \neq \emptyset$.
Otherwise, if $\C_\sigma = \emptyset$, the 
$Z$-predecessor of $\sigma$ in $\cL$ either does not exist
or is disjoint from $\sigma$. Thus, in this case,
we have $\emptyset = \sigma \cap \tau = \sigma \cap S$.
\end{proof}

\subparagraph*{The search algorithm.}
For a site $s \in S$, we define the 
\emph{neighborhood} $N(s)$ of $s$ as all 
cells in $\G$ with side length
$2^{\lfloor \log_2{r_s}\rfloor}$ that 
intersect $D_s$. The neighborhood
will be used to approximate $D_s$ for the
range search in the quadtrees.

\begin{lemma}\label{lem:neighbor_size}
There is a constant $\beta$ such that
$|N(s)| \leq \beta$ for all $s \in S$.
\end{lemma}

\begin{proof}
We have 
$r_s/2 < 2^{\lfloor \log_2{r_s}\rfloor}$, and
a $5\times5$ grid with cells of 
side length $r_s/2$ covers $D_s$, no 
matter where $s$ lies; 
see Figure~\ref{fig:constantneighborhood}.
Thus, the lemma holds with $\beta = 25$.
\end{proof}
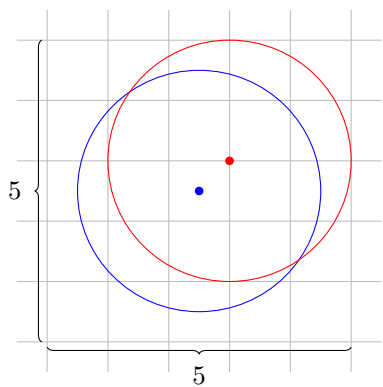
\begin{figure}
\center
\begin{tikzpicture}[scale=0.8]
\foreach \x in {-2.5,...,2.5}
{
  \draw[gray!50!white] (-3,\x) -- (3 , \x);
  \draw[gray!50!white] (\x,-3) -- (\x, 3);
}


\coordinate (center1) at (0,0);
\coordinate (moved1) at ($ (center1) +(2,0)$); 
\node[ourblue,circle, inner sep = 1, draw, fill] at (center1){};
\node[ourblue, draw] at (center1) [circle through =  {(moved1)}]{};
\coordinate (center2) at (0.5,0.5);
\coordinate (moved2) at ($ (center2)+ (2,0)$); 
\node[ourred, circle, inner sep = 1, draw, fill] at (center2){};
\node[ourred, draw] at (center2) [circle through =  {(moved2)}]{};

\draw[decoration={brace,raise=2pt},decorate] (2.5,-2.5) -- node[below=6pt] {$5$}  (-2.5,-2.5);
\draw[decoration={brace,raise=2pt},decorate] (-2.5,-2.5) -- node[left=6pt] {$5$}  (-2.5,2.5);
\end{tikzpicture}
\caption{The neighborhood of a site has constant 
size}\label{fig:constantneighborhood}
\end{figure}

We now show that a linearized compressed quadtree for
each canonical interval can be found without logarithmic
overhead.

\begin{lemma}
\label{lem:lin_quad}
We can compute for each
$v \in B$ the linearized quadtree $\cL_v$
for the sites in $\cD_v$ in $O(n \log n)$ time.
\end{lemma}

\begin{proof}
For each $v \in B$, we build the 
compressed quadtree $\C_v$ for  
$\cD_v$, as follows:  at the root, 
we compute the compressed
quadtree $\C$ for $S$ in 
$O(n \log n)$ time~\cite{BuchinLoMoMu11,har-peled_geometric_2008}.
Then, we traverse $B$.
Given the compressed quadtree $\C_v$
for a node $v \in B$, we compute 
$\C_w$ for a child $w$ of $v$ 
as follows.
We do a postorder traversal of $\C_v$. 
In each leaf $\nu$ of $\C_v$, we check if the site $s$ in $\nu$
is in $\cD_w$, by testing $r_s$.
If not, we remove $\nu$; otherwise, we keep it. 
In each inner vertex $\nu$ of $\C_v$, we check if $\nu$ 
has any remaining children. If not, we remove $\nu$.
If $\nu$ has exactly one remaining child that is 
not a compressed vertex, we mark $\nu$ as compressed and continue.
If the only remaining child of $\nu$ is compressed, we 
remove this child, connect $\nu$
to its grandchild, and mark $\nu$ as compressed.
This takes $O(|\C_v|)$ time and gives $\C_w$.

Once all the compressed quadtrees $\C_v$ are available,
we traverse $B$ again to
find the linearized compressed quadtrees $\cL_v$ 
by a traversal of
each $\C_v$. The total time to find the $\cL_v$
is $O(n \log n + \sum_{v \in B} |\C_v|) =
O(n \log n)$, since $|C_v| = O(|\cD_v|)$, for all $v \in B$, and
$\sum_{v \in B} |\cD_v| = O(n \log n)$ by 
Lemma~\ref{lem:canonical}.
\end{proof}

Using the linearized compressed quadtrees, the
range searching problem can be solved by a 
batched predecessor search, using a single traversal
of $B$.

\begin{lemma}
\label{lem:cquadtree_search}
The range searching problem (\typeROne) can be solved
in $O(n \log n)$ time.
\end{lemma}

\begin{proof}
We apply Lemma~\ref{lem:lin_quad} to find the linearized
quadtree for every canonical interval in \(B\).
Remember that the queries in (\typeROne) are the complete set \(S\). We split
each query $s \in S$ into subqueries, by considering
the neighborhood $N(s)$ of $s$.
Let
$\Q' = \bigcup_{s \in S} \big\{(\sigma, s) \mid \sigma \in N(s)\big\}$ 
be the set of \emph{split queries}. The purpose of the 
split queries is to approximate the associated disks for the
query sites by cells from the hierarchical grid.
By Lemma~\ref{lem:neighbor_size}, $|\Q'| = O(n)$.

We now perform range queries for the cells in the 
split queries.
For this, we first sort the elements of $\Q'$ in the $Z$-order of 
their first components, in $O(n \log n)$ time.
Next, we distribute the split queries along
their canonical paths in $B$.
For each $v \in B$, let $\Q'_v$ be the sorted
sublist of queries in $\Q'$ (in the $Z$-order
of the first component) that have $v$ on their
canonical path. By Lemmas~\ref{lem:canonical} 
and~\ref{lem:neighbor_size}, we
have $\sum_{v \in B} |\Q'_v| = O(n \log n)$. To 
find the lists $\Q'_v$ for all $v \in B$ in $O(n \log n)$ time,
we perform a pre-order traversal of $B$, computing the
lists for the children from the lists of the parents.
More precisely, given the sorted list $\Q'_v$ for a node $v \in B$,
we can find the sorted list $\Q'_w$ for a child 
$w$ of $v$ in time $O(|\Q'_v|)$ by scanning $\Q'_v$
from left to right and by copying the elements that
also appear in $\Q'_w$.
Finally, we distribute the split queries into their canonical nodes.
The canonical nodes of a query are children of the nodes
on its canonical path. Thus, we can find for each
$v \in B$ the sorted list $\Q''_v$
of split queries with $v$ as a canonical node
as follows: we iterate over all
non-root nodes $v \in B$, and we scan the list $\Q'_w$ of 
the parent node $w$ of $v$. We copy all queries
that have $v$ as a canonical node from $\Q'_w$ into $\Q''_v$.
This takes $O(n \log n)$ time.

Next, we iterate over all $v \in B$, and we merge the
lists $\Q''_v$ with the lists $\cL_v$, 
in $Z$-order. This takes 
$O\big(\sum_{v \in B} |\cL_v| + |\Q''_v|\big) = O(n \log n)$ time.
By Lemma~\ref{lem:quadtreeorder}, we obtain for each \((\sigma,s) \in \Q''_v\) a cell \(\tau_v^{\sigma,s}\).
If \(\sigma \cap \tau_v^{\sigma,s} \neq \emptyset\) we know that \(\sigma \cap \cD_v = \tau_v^{\sigma,s}\cap \cD_v\).
Since these sites are all from \(\cD_v\) they all have radius at least \(r_s/2\). 
We can find all these sites in \(O(k)\) time, where \(k\) is the
output size.
If $k > \alpha$, we stop and report $\sigma$ as
a square with many sites of large radius.\footnote{Note that here the radii are \(\geq r_s/2\) inside of the cells \(\sigma\).
This might be larger than the value \(2^{\lceil \log_2 r_s\rceil}/4\) needed by (\typeROne).
But still, if there are more than \(\alpha\) sites in \(\sigma\), we still have a triangle in a square. Otherwise we will later determine that each disk contains few sites of radius at least \(r_s/2\).}

Otherwise, we use the sites in $\sigma$ to accumulate
the sites for the query disk $D_s$. 
This we can do by considering all canonical nodes of \(s\) and for each cell \(\sigma\) iterate over the sites contained in \(\sigma\).
In each such cell there are at most \(\alpha\) sites. For each site \(t \in \sigma\) we can check if \(t\in D_s\). 
If we find a query disk $D_s$
with more than $\alpha$ sites of large radius,
we stop and report its enclosing square with many sites of large radius.\footnote{\(r=2r_s\) is the side length of the enclosing square, the radii are at least \(r/4\) as desired.}
Otherwise, 
for each $s \in S$, we have found the at most $\alpha$ sites 
of radius at least $r_s/2$ in
$D_s$. 
The whole algorithm takes $O(n \log n)$ time.
\end{proof}

\subsection{Queries of Type (\typeRTwo)}
\label{sec:RangeR2}

We use the tree structure of the 
canonical intervals (i) to construct quickly the search structures for 
each canonical interval; and (ii) to
solve all queries for a canonical interval in one batch.
We 
exploit the connection
between planar disks and
three-dimensional polytopes.
Let $U = \big\{(x,y,z) \mid x^2 + y^2 = z \big\}$ be 
the three-dimensional unit paraboloid. 
For a site $s \in S$, the 
lifted site $\hat{s}$ is 
the vertical projection of $s$ onto $U$.
Each disk $D_s$ is transformed into an upper halfspace 
$\widehat{D}_s$, 
so that the projection of 
$\widehat{D}_s \cap U$ onto the $xy$-plane is the set
$\R^2 \setminus D_s$;\footnote{This halfspace is bounded  by the plane 
\(z=2x_sx-x_s^2+2y_sy-y_s^2+r_s^2\), where \(s=(x_s, y_s)\).}
see Figure~\ref{fig:lifteddisks}.
The union of a set of disks in $\R^2$ 
corresponds to the intersection of the lifted upper 
halfspaces in $\R^3$.

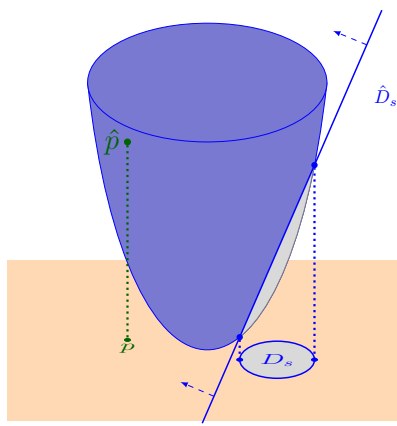
\begin{figure}
\centering
\scalebox{0.7}{
\makeatletter 
\tikzoption{canvas is xy plane at z}[]{%
  \def\tikz@plane@origin{\pgfpointxyz{0}{0}{#1}}%
  \def\tikz@plane@x{\pgfpointxyz{1}{0}{#1}}%
  \def\tikz@plane@y{\pgfpointxyz{0}{1}{#1}}%
  \tikz@canvas@is@plane
}

\tikzoption{at fix}[]{%
\def\tikz@plane@origin{\pgfpointxyz{-2}{1}{5}}%
  \def\tikz@plane@x{\pgfpointxyz{-7}{0}{0}}%
  \def\tikz@plane@y{\pgfpointxyz{0}{6}{0}}%
  \def\tikz@plane@z{\pgfpointxyz{0}{0}{-1}}%
  }

\tikzset{use path/.code=\tikz@addmode{\pgfsyssoftpath@setcurrentpath#1}}
\makeatother
\begin{tikzpicture}
[scale=0.75,
x={(1,0)},z={(0,1)}, y={(0,0.5)},
>=latex, 
point/.style={draw, circle, inner sep=1pt, fill}
]
\coordinate (bl) at (-5,-3.1,0);
\coordinate (br) at (-5,5,0);
\coordinate (tl) at (5,-3.1,0);
\coordinate (tr) at (5,5,0);

\coordinate (center) at ($(bl)!0.5!(tr)$);

\coordinate (paraleft) at (-3,0,7);
\coordinate (pararight) at (3,0,7);
\path[name path= paraboloid, save path =\para] (paraleft) .. controls (-2,0,-2) and (2,0,-2) .. (pararight);

\draw[fill=orange, fill opacity=0.3, draw opacity=0] (bl)--(br)-- (tr) -- (tl) -- cycle;

\begin{scope}[canvas is xy plane at z=0,transform shape]
\coordinate (left) at (-5,0);
\coordinate (right) at (5,0);
\coordinate (c1) at ($(left)!0.55!(right)$);
\coordinate (c2) at ($(left)!0.8!(right)$);
\node[thick, draw, name path = circle,circle through = (c1), ourblue, fill=gray!30!white] at ($(c1)!0.5!(c2)$){};
\node[ourblue] at ($(c1)!.5!(c2)$){\huge$D_s$};
\node[point,ourdarkgreen, inner sep=1.5pt, label={[ourdarkgreen]below:\huge$p$}](t) at (-2,1,0){};
\path[name path=horiz] (-5,0,0) -- (5,0,0);
\path[name intersections =  {of = horiz and circle, name = circint}];

\node[point, inner sep=1.5pt, ourblue] at (circint-1) {};
\node[point, inner sep=1.5pt, ourblue] at (circint-2) {};
\end{scope}

\path[name path = vert1 ] (circint-1) -- ($(circint-1) + (0,0,5)$);
\path[name path = vert2 ] (circint-2) -- ($(circint-2) + (0,0,5)$);
\path[name intersections = {of = vert1 and paraboloid, name= vertint1}];
\path[name intersections = {of = vert2 and paraboloid, name= vertint2}];
\draw[dotted, very thick,ourblue] (circint-1) -- (vertint1-1);
\draw[dotted, very thick, ourblue] (circint-2) -- (vertint2-1);
\coordinate (leftint) at ($(vertint1-1)!1.5!(vertint2-1)$);
\coordinate (rightint) at ($(vertint2-1)!1.9!(vertint1-1)$);
\coordinate (rightint1) at ($(vertint2-1)!2!(vertint1-1)$);
\coordinate (leftoff) at (-6,0,0);
\coordinate (rightoff) at (4,0,0);
\coordinate (bll) at ($(leftint) + (leftoff)$);
\coordinate (tll) at ($(rightint) + (leftoff)$);
\coordinate (brl) at ($(leftint) + (rightoff)$);
\coordinate (trl) at ($(rightint) + (rightoff)$);

\node[ourblue, above right = of vertint1-1] {$\hat{D}_s$};

\draw[draw = ourblue,fill=gray!50!ourblue!70!white][use path=\para];

\begin{scope}
\clip (leftint) -- (rightint) --(trl) -- (brl);
\draw[draw =gray,fill=gray!30!white ][use path=\para];
\end{scope}

 \begin{scope}[canvas is xy plane at z=7, transform shape, rotate=-1.5]
\draw[ourblue, fill=gray!50!ourblue!70!white] ($(paraleft)!0.5!(pararight)$) circle (3);
 \end{scope}

\coordinate (arrowpointu) at ($(leftint)!0.95!(rightint) -(1,0)$);
\draw[ourblue, ->, dashed] ($(leftint)!(arrowpointu)!(rightint)$)--(arrowpointu);
\coordinate (arrowpointl) at ($(leftint)!.1!(rightint) -(1,0)$);
\draw[ourblue, ->, dashed] ($(leftint)!(arrowpointl)!(rightint)$)--(arrowpointl);
\node[point, ourblue] at (vertint1-1) {};
\node[point,ourblue] at (vertint2-1) {};

\begin{scope}[at fix, transform shape]
\node[point, inner sep=1.5, ourdarkgreen, label={[ourdarkgreen]left:\huge$\hat{p}$}, transform shape](lt) at (-2,1,5){};
\end{scope}
\draw[ourblue,thick] (leftint)--(rightint);
\draw[dotted, very thick, ourdarkgreen] (lt)-- (t);

\end{tikzpicture}
}
\caption{Lifting disks and points. For $\hat{D}$ only the 
bounding plane is shown.}
\label{fig:lifteddisks}
\end{figure}

\begin{lemma}
\label{lem:complicated}
The range searching problem (\typeRTwo) can be solved
in $O(n \log n)$ expected time.
\end{lemma}

\begin{proof}
For each  $v \in B$, we construct a three-dimensional representation
of the union of the disks in the canonical interval  $\cD_v$.
As explained above, this is the intersection $\E_v$ of the lifted 
three-dimensional halfspaces $\widehat{D}_s$, for $s \in \cD_v$.
The intersection of two three-dimensional convex polyhedra 
with a total of $m$ vertices 
can be computed in $O(m)$ time~\cite{Chazelle92,chan_simpler_2016}.
Therefore, we can construct all the polyhedra
$\E_v$, for $v\in B$, in overall $O(n\log n)$ time, by 
a bottom-up traversal of $B$ (by Lemma~\ref{lem:canonical},
the total number of vertices of these polyhedra is $O(n \log n)$).

For the query processing,
we compute a polytope $\widehat{Q}_v$ 
for each $v \in B$. The polytope $\widehat{Q}_v$ is
obtained by determining all the points $p$ that
appear in a query $(p, r_1, r_2)$ that has $v$ as a
canonical node, lifting those points $p$ to their 
three-dimensional representations $\hat{p}$, and taking the convex hull 
of the resulting three-dimensional point set.
The lifted query points all lie on the unit paraboloid $U$,
so every lifted query point appears as a vertex on $\widehat{Q}_v$.
To find all polytopes $\widehat{Q}_v$, for $v \in B$,
efficiently, we proceed as follows:
let $A$ be the three-dimensional point 
set obtained by taking all points that appear in a query
and by lifting them onto the unit paraboloid.
We compute the convex hull of $A$ in $O(n \log n)$ 
time. Then, for each $v \in B$, we find the convex hull 
of all lifted queries that have $v$ in their canonical path. 
This can be done in $O(n \log n)$ total expected time by a top-down
traversal of $B$. We already have the polytope for the
root of $B$. To compute the polytope for a child node, given that
the polytope for the parent node is available,
we use the fact that for any 
polytope $\E$ in $\R^3$ with $m$ vertices, we can compute the convex hull 
of any subset of the vertices of $\E$ in $O(m)$ expected 
time~\cite{ChazelleMu11}. 
Once we have for each $v \in B$ the convex hull of the lifted query points that
have $v$ on their canonical \emph{path}, we can compute for each $v \in B$
the polytope $\widehat{Q}_v$ that is the convex hull of the 
lifted query points that have $v$ as a 
canonical \emph{node}. For this, we consider the canonical path polytope 
stored at the parent node of $v$, and we again use the algorithm 
from~\cite{ChazelleMu11} to extract the convex hull for the 
lifted query points that
have $v$ as a canonical node.

Now that the polyhedra
$\widehat{Q}_v$ and the polytopes $\E_v$ are available, for all
$v \in B$, we can answer the query as follows:
for each node $v \in B$, 
we must check for vertices of $\widehat{Q}_v$ that 
do not lie inside of $\E_v$. These are exactly 
the vertices of $\widehat{Q}_v$ that are not 
vertices of $\widehat{Q}_v \cap \E_v$. 
As mentioned, the intersections $\widehat{Q}_v \cap \E_v$ 
can be found in linear 
time for each node $v \in B$, for a total time $O(n \log n)$, 
and once the intersection is available, we can easily
find
all vertices $\hat{p}$ of $\widehat{Q}_v$ that
are not vertices of \(\widehat{Q}_v\cap \E_v\) (e.g., using radix sort). 
If for any such intersection $\widehat{Q}_v \cap \E_v$, 
there is a lifted site $\hat{s} \in \widehat{Q}_v$  
that is not a vertex of $\widehat{Q}_v \cap \E_v$, we report $s$ as
the result of the range search. Otherwise, we report that the 
range search is unsuccessful.
\end{proof}

\section{Finding the Shortest Triangle in a Transmission Graph}
\label{sec:finding_the_smallest_triangle}

We extend 
Theorem~\ref{thm:transmission_triangle_unweighted}
to find the shortest triangle 
in  $T(S)$. As in Section~\ref{sec:disk_triangle_weighted}, 
we solve the decision problem:
given $W > 0$, does $T(S)$ have a directed triangle of perimeter 
at most $W$?
We set $\ell = W/\sqrt{27}$, and
call a site $s \in S$ \emph{large} if 
$r_s > \ell$. We let $S_\ell  \subseteq S$ be the 
set of all large sites.

\begin{lemma}\label{lem:find_small_triangles}
We can find a triangle in $T(S \setminus S_\ell)$ 
of perimeter at most $W$
in $O(n\log n)$ time, if it exists.
\end{lemma}

\begin{proof}
Any triangle in $T(S \setminus S_\ell)$ has 
perimeter at most $W$: consider a directed
triangle $stu$ in 
$T(S \setminus S_\ell)$
with $r_s \geq \max\{r_t, r_u\}$.
Then we have $t, u \in D_s$, so 
the triangle $stu$ lies in $D_s$.
Elementary calculus shows that 
a triangle of maximum perimeter in
$D_s$ must be equilateral 
with its vertices on $\partial D_s$,
so any triangle contained in $D_s$ has perimeter at most
$3\cdot \sqrt{3}\cdot r_s  \leq \sqrt{27}\cdot \ell =  W$.
We can find a triangle in $T'$ in $O(n\log n)$ time by 
Theorem~\ref{thm:unweightedtriangledir}.
\end{proof}

It remains to check for 
triangles of perimeter at most $W$ 
with at least one large vertex.
Some such triangles have to be considered 
individually, while the others can be handled efficiently 
in batch mode.
The following lemma shows that we may assume that
there are few edges from $S \setminus S_\ell$ to $S_\ell$.

\begin{lemma}
\label{lem:small_outdegree}
If $T(S)$ does not have a triangle of perimeter 
at most $W$, every site in $S_\ell$ has at most six incoming 
edges from $S \setminus S_\ell$.
Furthermore, in $O(n \log n)$ time, we can either find 
a triangle of perimeter at most $W$ in $T(S)$ or determine
for each site in $S_\ell$ all incoming edges from $S \setminus S_\ell$.
\end{lemma}

\begin{proof}
Suppose there is a square $\sigma$ in the plane with side length
$0 < r < 2\ell$ such that $\sigma$ contains more than $\alpha$ 
sites $s$ of radius \(r_s\geq r/4\). 
Then, Lemma~\ref{lem:disks_close_together} 
shows that $T(S)$ contains a triangle that lies in $\sigma$.
Specifically, since $\sigma$ has 
side length at most $r/6 = 2\ell/6$, the definition
of $\ell$ implies that the triangle has perimeter at most $W$. 
If there is no such square, it follows that there
is no site $s$ with $r_s \leq \ell$ such that $D_s$ contains
more than $\alpha$ sites of radius at least $r_s/2$, as 
otherwise $s$ could be enclosed by 
a square of side length $2r_s \leq 2\ell$ that contains 
many sites of large radius.
Thus, every small site $s$ has $O(1)$ outgoing edges to 
sites with radius at least $r_s/2$.
In particular, there are $O(n)$ edges from small sites to large sites.
We can use a suitable variant of (\typeROne) so that in 
$O(n \log n)$ time we can EITHER find a square of
side length $0 < r < 2\ell$ that contains more than $\alpha$ sites $s$ of 
radius $r_s \geq r/4$r OR determine that for every site $s$ 
with $r_s \leq \ell$, there are at most $\alpha$ sites in $D_s$ of radius
at least $r_s/2$. Furthermore, in the second case, we  explicitly
get all sites of radius at least $r_s/2$ in each $D_s$ with $r_s \leq \ell$

Thus, if the second case applies, we obtain 
all edges from $S \setminus S_\ell$ to $S_\ell$. Suppose there is a 
large site $s$ of indegree at least $7$. Then, there must be two
sites $t,u \in S \setminus S_\ell$ such that the angle between
the edges $ts$ and $us$ 
is less than $\pi/3$. Thus, the distance $|tu|$ is less than the 
maximum of $|ts|$ and $|us|$, so there is a directed edge with 
endpoints $t$ and $u$ and the sites $s, t, u$ form a triangle 
of perimeter at most $3\ell \leq W$.
\end{proof}

Next, we want to limit the number of relevant edges between large sites.
For this, we subdivide the plane with a grid $G$ of side 
length $\ell/\sqrt{2}$.
Then, we have the following:

\begin{lemma}\label{lem:constant_large_sites}
A triangle contained in a cell $\sigma \in G$ 
has perimeter at most $W$. 
If there is no triangle in $\sigma$, then $\sigma$ 
contains $O(1)$ large sites.
We can check for such triangles in $O(n\log n)$ overall expected time.
\end{lemma}

\begin{proof}
The maximum perimeter of a triangle contained in 
$\sigma$ is $(1 + \sqrt{2})\ell < W$.
Furthermore, if there are at least three large sites in $\sigma$, 
these large sites form a triangle, since the disk of a large 
site covers $\sigma$.
By applying Theorem~\ref{thm:unweightedtriangledir} 
to the induced subgraph in each
cell of $G$, we can find such a triangle in $O(n\log n)$ total 
expected time.
\end{proof}

We define the neighborhood $N(\sigma)$ of a cell $\sigma \in G$ as the 
$5\times 5$ block of cells centered at $\sigma$.
Let $t$ be a site and $\sigma$ the cell containing $t$, 
then the neighborhood $N(t)$ of $t$ are all sites 
contained in $N(\sigma)$.
Since the side length of a grid cell 
is $W/3\sqrt{6}$, each triangle of perimeter at 
most $W$ is completely contained in the neighborhood of some cell.

\begin{lemma}\label{lem:remaining_mixed}
We can check the remaining triangles in $O(n)$ overall time.
\end{lemma}

\begin{proof}
Consider a remaining triangle  
$sut$ with $r_t \geq \max\{r_u, t_s\}$.
Then, $t \in S_\ell$, and $s, t, u$ all lie in $N(t)$.
By Lemma~\ref{lem:constant_large_sites}, there 
are $O(1)$ large candidates for $u$ in $N(t)$,
and by Lemma~\ref{lem:small_outdegree},
there are $O(1)$ small candidates for $u$.
Having fixed a $t$ and a possible candidate $u$, 
we iterate over all $s\in N(t)$ and check if $s$, $u$, and 
$t$ form a triangle with weight at most $W$.
Every site $s$ is contained in $O(1)$ grid neighborhoods, 
and since there are $O(1)$ candidate pairs in each grid
neighborhood, $s$ participates in $O(1)$ explicit checks.
The result follows.
\end{proof}
The following theorem summarizes the considerations
in this section.
\begin{theorem}
It takes $O(n\log n)$ expected time 
to find the shortest triangle in 
a transmission graph.
\end{theorem}

\begin{proof}
We already saw that there is an $O(n\log n)$ time 
decision algorithm for the problem. 
As in Theorem~\ref{thm:shortesttriangle},
the result follows from an application of 
Chan's randomized optimization 
technique~\cite{Chan99} (restated in Lemma~\ref{lem:chan}).
\end{proof}

\section{Conclusion}

Once again, disk graphs and transmission 
graphs prove to be
a simple yet powerful graph model where
difficult algorithmic problems admit faster solutions.
It would be interesting to find a \emph{deterministic}
$O(n \log n)$ time algorithm for finding
a shortest triangle in a disk graph. Currently, 
we are working on extending our results to the girth problem 
in transmission graphs; can we find an equally simple and 
efficient algorithm as for disk graphs?

\bibliography{literature}

\newcommand{\SortNoop}[1]{}
\begin{thebibliography}{10}

\bibitem{AlonYuZw97}
Noga Alon, Raphael Yuster, and Uri Zwick.
\newblock Finding and counting given length cycles.
\newblock {\em Algorithmica}, 17(3):209--223, 1997.

\bibitem{dBCvKO}
Mark {\SortNoop{Berg}}de~Berg, Otfried Cheong, Marc van Kreveld, and Mark~H.
  Overmars.
\newblock {\em Computational Geometry: Algorithms and Applications}.
\newblock Springer-Verlag, third edition, 2008.

\bibitem{BuchinLoMoMu11}
Kevin Buchin, Maarten L{\"{o}}ffler, Pat Morin, and Wolfgang Mulzer.
\newblock Preprocessing imprecise points for {D}elaunay triangulation:
  Simplified and extended.
\newblock {\em Algorithmica}, 61(3):674--693, 2011.

\bibitem{BuchinMu11}
Kevin Buchin and Wolfgang Mulzer.
\newblock Delaunay triangulations in ${O}(\text{sort}(n))$ time and more.
\newblock {\em J. ACM}, 58(2):6:1--6:27, 2011.

\bibitem{Chan99}
Timothy~M. Chan.
\newblock Geometric applications of a randomized optimization technique.
\newblock {\em Discrete Comput. Geom.}, 22(4):547--567, 1999.

\bibitem{chan_simpler_2016}
Timothy~M. Chan.
\newblock A simpler linear-time algorithm for intersecting two convex polyhedra
  in three dimensions.
\newblock {\em Discrete Comput. Geom.}, 56(4):860--865, December 2016.

\bibitem{ChangLu13}
Hsien-Chih Chang and Hsueh-I Lu.
\newblock Computing the girth of a planar graph in linear time.
\newblock {\em SIAM J. Comput.}, 42(3):1077--1094, 2013.

\bibitem{Chazelle92}
Bernard Chazelle.
\newblock An optimal algorithm for intersecting three-dimensional convex
  polyhedra.
\newblock {\em SIAM J. Comput.}, 21(4):671--696, 1992.

\bibitem{ChazelleMu11}
Bernard Chazelle and Wolfgang Mulzer.
\newblock Computing hereditary convex structures.
\newblock {\em Discrete Comput. Geom.}, 45(4):796--823, 2011.

\bibitem{CormenLeRiSt09}
Thomas~H. Cormen, Charles~E. Leiserson, Ronald~L. Rivest, and Clifford Stein.
\newblock {\em Introduction to algorithms}.
\newblock MIT Press, third edition, 2009.

\bibitem{EvansGaLoPo16}
William~S. Evans, Mereke van Garderen, Maarten L{\"o}ffler, and Valentin
  Polishchuk.
\newblock Recognizing a {DOG} is hard, but not when it is thin and unit.
\newblock In {\em Proc. 8th FUN}, pages 16:1--16:12, 2016.

\bibitem{LeGall14}
Fran{\c{c}}ois {\SortNoop{Gall}}Le~Gall.
\newblock Powers of tensors and fast matrix multiplication.
\newblock In {\em Proc. 39th Internat. Symp. Symbolic and Algebraic Comput.
  (ISSAC)}, pages 296--303, 2014.

\bibitem{har-peled_geometric_2008}
Sariel Har-Peled.
\newblock {\em Geometric approximation algorithms}.
\newblock American Mathematical Society, 2008.

\bibitem{ItaiRo78}
Alon Itai and Michael Rodeh.
\newblock Finding a minimum circuit in a graph.
\newblock {\em SIAM J. Comput.}, 7(4):413--423, 1978.

\bibitem{KaplanMuRoSe15}
Haim Kaplan, Wolfgang Mulzer, Liam Roditty, and Paul Seiferth.
\newblock Spanners and reachability oracles for directed transmission graphs.
\newblock In {\em Proc. 31st Int. Sympos. Comput. Geom. (SoCG)}, pages
  156--170, 2015.

\bibitem{KaplanMuRoSe18}
Haim Kaplan, Wolfgang Mulzer, Liam Roditty, and Paul Seiferth.
\newblock Spanners for directed transmission graphs.
\newblock {\em SIAM J. Comput.}, 47(4):1585--1609, 2018.

\bibitem{lacki_min-cuts_2011}
Jakub {\L}{\c a}cki and Piotr Sankowski.
\newblock Min-cuts and shortest cycles in planar graphs in ${O}(n \log\log n)$
  time.
\newblock In {\em Proc. 19th Annu. European Sympos. Algorithms (ESA)}, pages
  155--166, 2011.

\bibitem{MozesNiNuWe18}
Shay Mozes, Kirill Nikolaev, Yahav Nussbaum, and Oren Weimann.
\newblock Minimum cut of directed planar graphs in ${O}(n \log \log n)$ time.
\newblock In {\em Proc. 29th Annu. ACM-SIAM Sympos. Discrete Algorithms
  (SODA)}, pages 477--494, 2018.

\bibitem{PapadimitriouYa81}
Christos~H. Papadimitriou and Mihalis Yannakakis.
\newblock The clique problem for planar graphs.
\newblock {\em Inform. Process. Lett.}, 13(4/5):131--133, 1981.

\bibitem{Polishchuk17}
Valentin Polishchuk.
\newblock Personal communication.
\newblock 2017.

\bibitem{RodittyVa11}
Liam Roditty and Virginia Vassilevska~Williams.
\newblock Minimum weight cycles and triangles: Equivalences and algorithms.
\newblock In {\em Proc. 52nd Annu. IEEE Sympos. Found. Comput. Sci. (FOCS)},
  pages 180--189, 2011.

\bibitem{WillardLu85}
Dan~E. Willard and George~S. Lueker.
\newblock Adding range restriction capability to dynamic data structures.
\newblock {\em J. ACM}, 32(3):597--617, 1985.

\bibitem{WiWi18}
Virginia~Vassilevska Williams and R.~Ryan Williams.
\newblock Subcubic equivalences between path, matrix, and triangle problems.
\newblock {\em J. ACM}, 65(5):27:1--27:38, 2018.

\bibitem{Yu15}
Huacheng Yu.
\newblock An improved combinatorial algorithm for {B}oolean matrix
  multiplication.
\newblock In {\em Proc. 42nd Internat. Colloq. Automata Lang. Program.
  (ICALP)}, pages 1094--1105, 2015.

\bibitem{Yuster11}
Raphael Yuster.
\newblock A shortest cycle for each vertex of a graph.
\newblock {\em Inform. Process. Lett.}, 111(21-22):1057--1061, 2011.

\end{thebibliography}
\end{document}